\documentclass[11pt]{amsart}
\usepackage{amsfonts,amssymb,amsmath,amsthm,mathrsfs}
\usepackage{url}
\usepackage{enumerate}
\usepackage{txfonts}

\newtheorem{thrm}{Theorem}[section]
\newtheorem{lem}[thrm]{Lemma}
\newtheorem{prop}[thrm]{Proposition}
\newtheorem{cor}[thrm]{Corollary}
\theoremstyle{definition}
\newtheorem{definition}[thrm]{Definition}
\newtheorem{remark}[thrm]{Remark}
\numberwithin{equation}{section}

\author{Lin Zhang and Junde Wu}
\address{Department of Mathematics\\
Zhejiang University\\Hangzhou\\
People's Republic of China}
\email{godyalin@163.com, linyz@zju.edu.cn}

\begin{document}

\title[Dynamical matrices]{A Survey of Dynamical Matrices Theory}

\begin{abstract}
In this note, we survey some elementary theorems and proofs concerning dynamical
matrices theory. Some mathematical concepts and results involved in quantum information theory are reviewed.
A little new result on the matrix representation of quantum operation are obtained. And best separable
approximation for quantum operations is presented.
\end{abstract}
\maketitle

\tableofcontents
\pagebreak

\section{Introduction and preliminaries}
\emph{Positive linear maps} on some operator algebras are a very
important subject of both the mathematical and the physical
literature for several years. The images of positive operators
acting on a given Hilbert space under such a map are positive
operators acting on the same Hilbert space.  A map $\Phi$ is called
$k$-positive for some $k\in\mathbb{N}$ if the tensor product
$\Phi\otimes \mbox{Id}_{k}$ is positive. We call $\Phi$ is a
\emph{completely positive} (CP) when it is $k$-positive for any
$k\in\mathbb{N}$. Completely positive maps (CP maps, for short)
describe the dynamics of open quantum systems.
The structure of the set of CP maps is well understood due to the theorems of Stinespring \cite{Stinespring}, Kraus \cite{Kraus}, and Choi \cite{Choi}. Choi's theorem is also proved by another simple approach in \cite{Salgado}.\\

In this paper, only finite
dimensional complex vector spaces are considered. An column vector
in a complex vector space is denoted by $|\phi\rangle$, the symbol
$\phi$ is a label, while $|\cdot\rangle$ denotes that the object is
a complex column vector. This notation for complex vectors is called
\emph{Dirac notation}. Throughout the paper, $\dagger,\mbox{t}$ and $\ast$ stand for Hermitian conjugate,
transposition and complex conjugate, respectively, of matrices with respect to a given orthonormal basis. Given a vector\\
\indent $|\phi\rangle=[\phi_{1},\phi_{2},\ldots,\phi_{d}]^{\mbox{t}}$,\\ its dual is defined as\\
\indent $\langle\phi|=[\phi^{\ast}_{1}\ \phi^{\ast}_{2}\ \cdots\ \phi_{d}^{\ast}]\equiv(|\phi\rangle)^{\dagger}$.\\
Given the vectors $|\phi\rangle,|\varphi\rangle$, the inner product
between two vectors is denoted by $\langle\phi|\varphi\rangle$,
which is defined as follows:\\
\indent $\langle\phi|\varphi\rangle\equiv\sum_{i=1}^{d}\phi^{\ast}_{i}\varphi_{i}=[\phi^{\ast}_{1}\ \phi^{\ast}_{2}\ \cdots\ \phi^{\ast}_{d}][\phi_{1},\phi_{2},\ldots,\phi_{d}]^{\mbox{t}}$.\\
The \emph{norm} of a vector $|\phi\rangle$ is defined as
$\|\phi\|=\sqrt{\langle\phi|\phi\rangle}$. Unite vectors are those
vectors with unit norm. Two vectors are \emph{orthogonal} it they
have zero product. The \emph{outer product} of the given vectors
$|\phi\rangle$ and $|\varphi\rangle$ is given by
$$|\phi\rangle\langle\varphi|\equiv\left[\begin{array}{c}
                                                     \phi_{1} \\
                                                     \phi_{2} \\
                                                     \vdots \\
                                                     \phi_{d}
                                                   \end{array}
\right][\varphi^{\ast}_{1}\ \varphi^{\ast}_{2}\ \cdots\
\varphi^{\ast}_{d}]=\left[\begin{array}{cccc}
                      \phi_{1}\varphi^{\ast}_{1} & \phi_{1}\varphi^{\ast}_{2} & \cdots & \phi_{1}\varphi^{\ast}_{d} \\
                      \phi_{2}\varphi^{\ast}_{1} & \phi_{2}\varphi^{\ast}_{2} & \cdots & \phi_{2}\varphi^{\ast}_{d} \\
                      \vdots & \vdots & \vdots & \vdots \\
                      \phi_{d}\varphi^{\ast}_{1} & \phi_{d}\varphi^{\ast}_{2} & \cdots & \phi_{d}\varphi^{\ast}_{d}
                    \end{array}
\right].$$ A set of vectors $\{|v_{k}\rangle\}_{k=1}^n$ in a vector
space $\mathcal{V}$ is \emph{orthonormal} if the vectors are normalized and
orthogonal, that is, $\langle v_{i}|v_{j}\rangle=\delta_{ij}$. If,
in addition, $n=\dim\mathcal{V}$, this set of vectors form an orthonormal
basis for $\mathcal{V}$. Here we have a simple but useful fact that
$\sum_{k=1}^{n}|v_{k}\rangle\langle v_{k}|=I_{n}$ for given an
orthonormal basis $\{|v_{k}\rangle\}_{k=1}^n$ in a vector space $\mathcal{V}$.
This called the \emph{completeness relation}.\\

\emph{Quantum states} will now be introduced. A \emph{quantum system} is a physical system that obeys the laws of
quantum mechanics. Let us assume that we are given two quantum
systems. The first one is owned by Alice, and the second one by Bob.
The physical \emph{states} of Alice's system may be described by
states in a Hilbert space $\mathcal{H}_{A}$ of dimension $d_{A}=N$, and
in Bob's system in a Hilbert space $\mathcal{H}_{B}$ of dimension
$d_{B}=M$. The \emph{tensor product} is a ubiquitous mathematical operation which can be used to combine vector spaces to form a larger vector space.
Given two vector spaces $\mathcal{V}$ and $\mathcal{W}$, we can combine them to form the vector space $\mathcal{V}\otimes\mathcal{W}$, with $\dim(\mathcal{V}\otimes\mathcal{W})=\dim(\mathcal{V})\times\dim(\mathcal{W})$. The bipartite quantum system is then described by vectors in
the tensor-product of the two spaces
$\mathcal{H}=\mathcal{H}_{A}\otimes\mathcal{H}_{B}$, and
$\mbox{dim}(\mathcal{H})=d_{A}d_{B}$. A \emph{pure state} of dimension $d$ can be represented by a
$d$-dimensional complex unit vector $|\psi\rangle$. For real $\theta$, the vectors $|\psi\rangle$ and $e^{i\theta}|\psi\rangle$
 represent the same state. More generally, a $d$-dimensional quantum state is represented by a $d\times d$ complex matrix $\rho$, also called
 a \emph{density matrix}, which is a non-negative linear operator, acting on a complex Hilbert space $\mathcal{H}$, with trace 1.
 A pure state can be represented either by its state vector $|\psi\rangle$, or by its density matrix $\rho=|\psi\rangle\langle\psi|$.
 States which are not pure are called \emph{mixed states}. A simple test for whether a state $\rho$ is pure or mixed is to take the trace of $\rho^{2}$:
 $\mbox{tr}(\rho^{2})=1$ if $\rho$ is pure and $\mbox{tr}(\rho^{2})<1$ if $\rho$ is mixed. A mixed state can be expressed as a mixture of pure states in many different ways.\\

Suppose that $|v\rangle\in\mathcal{V},|w\rangle\in\mathcal{W}$. The vector $|v\rangle\otimes|w\rangle\in \mathcal{V}\otimes\mathcal{W}$. The vector $|v\rangle\otimes|w\rangle$ is computed as follows:
$$|v\rangle\otimes|w\rangle=\left[\begin{array}{c}
                                   w_{1}|v\rangle \\
                                   \vdots \\
                                   w_{k}|v\rangle \\
                                   \vdots \\
                                   w_{n}|v\rangle
                                 \end{array}
\right]\ \mbox{if}\ |w\rangle=\left[\begin{array}{c}
                                   w_{1} \\
                                   \vdots \\
                                   w_{k} \\
                                   \vdots \\
                                   w_{d_{B}}
                                 \end{array}
\right]\ \mbox{and}\ |v\rangle=\left[\begin{array}{c}
                                   v_{1} \\
                                   \vdots \\
                                   v_{k} \\
                                   \vdots \\
                                   v_{d_{A}}
                                 \end{array}
\right].$$
Similarly, the tensor product of two given matrices will be explained as follows:
with the orthonormal bases $\{|m\rangle\}(m=1,\ldots,d_{A})$ of
$\mathcal{H}_{A}$ and $\{|\mu\rangle\}(\mu=1,\ldots,d_{B})$ of
$\mathcal{H}_{B}$, respectively, the orthonormal basis of
$\mathcal{H}$ can be described as $\{|m\rangle\otimes
|\mu\rangle\equiv|m\mu\rangle\}(m=1,\ldots,d_{A};\mu=1,\ldots,d_{B})$ (throughout the present paper, Roman indices correspond to the subsystem $A$ and Greek indices to the subsystem $B$.)
for which two types of ordering are very important such as:
\begin{enumerate}[(i)]
\item Ordering of type-I:\\
$\{|11\rangle,|21\rangle,\ldots,|d_{A}1\rangle;\ldots;|1\mu\rangle,|2\mu\rangle,\ldots,|d_{A}\mu\rangle;\ldots;|1d_{B}\rangle,|2d_{B}\rangle,\ldots,|d_{A}d_{B}\rangle\}.$
\item Ordering of type-II:\\
$\{|11\rangle,|12\rangle,\ldots,|1d_{B}\rangle;\ldots;|m1\rangle,|m2\rangle,\ldots,|md_{B}\rangle;\ldots;|d_{A}1\rangle,|d_{A}2\rangle,\ldots,|d_{A}d_{B}\rangle\}.$
\end{enumerate}

$\mathscr{B}(\mathcal{H})$, $\mathscr{B}(\mathcal{H}_{A})$ and $\mathscr{B}(\mathcal{H}_{B})$ means that the set of all bounded linear operators on $\mathcal{H},\mathcal{H}_{A}$
and $\mathcal{H}_{B}$, respectively. If $X\in\mathscr{B}(\mathcal{H}_{A})$ and $Y\in\mathscr{B}(\mathcal{H}_{B})$,
then $X\otimes Y\in\mathscr{B}(\mathcal{H})$. Suppose that the matrix-representations $X\equiv[x_{mn}]$ and $Y\equiv[y_{\mu\nu}]$
for $X$ and $Y$ with respect to the given orthonormal bases $\{|m\rangle\}_{m=1}^{d_{A}}$ and $\{|\mu\rangle\}_{\mu=1}^{d_{B}}$ are given, respectively.
Then there are several different matrix-representations of $X\otimes Y$ with respect to the corresponding orthonormal bases of different orderings.
For the ordering of type-I, the matrix representation of $X\otimes Y$ is\\
\indent $X\otimes Y\equiv\left[\begin{array}{cccc}
                          y_{11}X & y_{12}X & \cdots & y_{1d_{B}}X \\
                          y_{21}X & y_{22}X & \cdots & y_{2d_{B}}X \\
                          \vdots & \vdots & \vdots & \vdots \\
                          y_{d_{B}1}X & y_{d_{B}2}X & \cdots & y_{d_{B}d_{B}}X
                        \end{array}
\right]$;\\
while for the ordering of type-II, the matrix representation of $X\otimes Y$ is\\
\indent $X\otimes Y\equiv\left[\begin{array}{cccc}
                          x_{11}Y & x_{12}Y & \cdots & x_{1d_{A}}Y \\
                          x_{21}Y & x_{22}Y & \cdots & x_{2d_{A}}Y \\
                          \vdots & \vdots & \vdots & \vdots \\
                          x_{d_{A}1}Y & x_{d_{A}2}Y & \cdots & x_{d_{A}d_{A}}Y\\
                        \end{array}
\right]$.\\
The ordering of type-I will be employed throughout the present paper if unspecified.
For tensor product, we have the following rules: given two matrices $S$ and $T$ acting on vector spaces $\mathcal{V}$ and $\mathcal{W}$, respectively, vectors $|x\rangle\in\mathcal{V}$ and $|y\rangle\in\mathcal{W}$, then\\
\indent $(S\otimes T)(|v\rangle\otimes|w\rangle)=(S|x\rangle)\otimes (T|y\rangle),\mbox{Tr}(S\otimes T)=\mbox{tr}(S)\mbox{tr}(T),(S\otimes T)^{\dagger}=S^{\dagger}\otimes T^{\dagger}$.\\
If $X,Y$ act also on $\mathcal{V,W}$ respectively, we have $(S\otimes T)(X\otimes Y)=SX\otimes TY$. Obviously, tensor product is a bi-linear map.\\

The description of \emph{subsystems} of a composite quantum system is provided by the \emph{reduced density operator}, which is so useful as to be virtually indispensable in the analysis of composite quantum systems.
Suppose we have physical systems $A$ and $B$, whose state is described by a density operator $\rho_{AB}$. The state space of the composite quantum system $AB$ is denoted by $\mathscr{D}(\mathcal{H})$, similarly, $\mathscr{D}(\mathcal{H}_{A})$ for subsystem $A$ and $\mathscr{D}(\mathcal{H}_{B})$ for subsystem $B$. The reduced density operator for system $A$ is defined by $\mbox{Tr}_{B}(\rho_{AB})\equiv\rho_{A}$,
where $\mbox{Tr}_{B}$ is a map of operators know as the \emph{partial trace} over system $B$. The partial trace is defined by\\
\indent $\mbox{Tr}_{B}(|a_{1}\rangle\langle a_{2}|\otimes |b_{1}\rangle\langle b_{2}|)=|a_{1}\rangle\langle a_{2}|\mbox{tr}(|b_{1}\rangle\langle b_{2}|)$,\\
where $|a_{1}\rangle$ and $|a_{2}\rangle$ are any two vectors in the
state space of $A$, and $|b_{1}\rangle$ and $|b_{2}\rangle$ are any
two vectors in the state space of $B$. The trace operation appearing
on the right hand side is the usual trace operation for system $B$,
so $\mbox{tr}(|b_{1}\rangle\langle b_{2}|)=\langle
b_{2}|b_{1}\rangle$. In fact,
$\mbox{Tr}_{A}=\mbox{tr}\otimes\mbox{Id}_{B}$,
$\mbox{Tr}_{B}=\mbox{Id}_{A}\otimes\mbox{tr}$ and
$\mbox{Tr}=\mbox{tr}\otimes\mbox{tr}$. We have defined the partial
trace operation only on a special subclass of operators on $AB$.
More generally, for any matrix $Z$ acting on
$\mathcal{H}_{A}\otimes\mathcal{H}_{B}$, we have a block
construction on $Z$:
$Z=[Z_{\mu\nu}](\mu,\nu=1,\ldots,d_{B}=\dim\mathcal{H}_{B})$, where
each $Z_{\mu\nu}$ is a scalar matrix of size $d_{A}\times
d_{A}(d_{A}=\dim\mathcal{H}_{A})$. Therefore\\
\indent $Z=\sum_{\mu,\nu=1}^{d_{B}}Z_{\mu\nu}\otimes|\mu\rangle\langle\nu|$.\\
Now the partial trace over system $A$ is provided by\\
\indent $\mbox{Tr}_{A}(Z)=\sum_{\mu,\nu=1}^{d_{B}}\mbox{tr}(Z_{\mu\nu})|\mu\rangle\langle\nu|\equiv[\mbox{tr}(Z_{\mu\nu})]$,\\
while the partial trace over system $B$ is given by\\
\indent $\mbox{Tr}_{B}(Z)=\sum_{\mu,\nu=1}^{d_{B}}Z_{\mu\nu}\mbox{tr}(|\mu\rangle\langle\nu|)=\sum_{\mu=1}^{d_{B}}Z_{\mu\mu}$.\\
The partial trace over the composite quantum system
$AB$ is \\
\indent $\mbox{Tr}(Z)=\sum_{\mu=1}^{d_{B}}\mbox{tr}(Z_{\mu\mu})$.\\

The quantum operations formalism is a general tool for describing
the evolution of quantum systems in a wide variety of circumstances,
including stochastic changes to quantum states. A simple example of
a state change in quantum mechanics is the unitary evolution
experienced by a closed quantum system. The final state of the
system is related to the initial state by a unitary transformation
$U$,\\
\indent $\rho\rightarrow \mathcal{E}(\rho)=U\rho U^{\dagger}$.\\
Unitary evolution is not the most general type of state change
possible in quantum mechanics. Other state changes, described
without unitary transformations, arise when a quantum system is
coupled to an environment or when a measurement is performed on the
system. This formalism is described in detail by Kraus. In this
formalism there is an input state and an output state, which are
connected by a map \\
\indent $\rho\rightarrow
\frac{\mathcal{E}(\rho)}{\mbox{tr}[\mathcal{E}(\rho)]}$.\\
This map
is determined by a \emph{quantum operation} $\mathcal{E}$, a linear,
trace-decreasing map that preserves positivity. The trace in the
denominator is included in order to preserve the trace condition
$\mbox{tr}(\rho)=1$. The most general form for $\mathcal{E}$ that is
physically reasonable, can be shown to be\\
\indent $\mathcal{E}(\rho)=\sum_{j}\Gamma_{j}\rho \Gamma_{j}^{\dagger}$.\\
The system operators $\Gamma_j$ , which must satisfy
$\sum_{j}\Gamma_{j}\Gamma_{j}^{\dagger}\leq I$, completely specify
the quantum operation. Formally, every quantum operation has to be
described mathematically by a completely positive complex-linear
mapping $\mathcal{E}$, which satisfies
$\mbox{tr}(\mathcal{E}(\rho))\leq 1$ for all state $\rho$. A quantum
operation is called \emph{quantum channel} if it is
trace-preserving.\\

 Given quantum operation $\mathcal{E},\mathcal{E}_{A},$ and $\mathcal{E}_{B}$ on corresponding bipartite quantum system with subsystems $A$ and $B$, subsystems $A$, and $B$, respectively, owing to Jamio{\l}kowski isomorphism, the notion of entanglement can be extended from quantum states to quantum operations. A quantum operation acting on two subsystems is said to be \emph{separable} if its action can be expressed in the Kraus form\\
\indent $\mathcal{E}(\cdot)=\sum_{k}(\Lambda_{k}^{A}\otimes\Lambda_{k}^{B})\cdot(\Lambda_{k}^{A}\otimes\Lambda_{k}^{B})^{\dagger}$,\\
where $\Lambda_{k}^{A}$ and $\Lambda_{k}^{B}$ are operators acting on each subsystem and they satisfy
that \\
\indent $\sum_{k}(\Lambda_{k}^{A}\otimes\Lambda_{k}^{B})^{\dagger}(\Lambda_{k}^{A}\otimes\Lambda_{k}^{B})\leq I_{A}\otimes I_{B}$.\\
Otherwise, it is entangled. When the equality is valid, there is a concept of \emph{separable quantum channel}.
\section{Vectorization and realignment of matrices}
\begin{definition} Representation of matrices as vectors on a
higher dimensional Hilbert space is called \emph{vectorization}. It
transforms a $p\times q$ matrix $G$ into $pq\times 1$ column vector
denoted by $|G\rangle\rangle$, this is done by ordering matrix
elements, i.e., by stacking the columns of $G$ to form a vector: for
example, with a $p\times q$ matrix $G=[g_{ij}]$, $|G\rangle\rangle$
is described as\\
\indent $
|G\rangle\rangle=\left[
                                \begin{array}{c}
                                  G(\cdot,1) \\
                                  \vdots \\
                                  G(\cdot,q) \\
                                \end{array}
                              \right]
,\ \mbox{where}\ G(\cdot,j)=\left[
                    \begin{array}{c}
                      g_{1j}\\
                      \vdots\\
                      g_{pj}\\
                      \end{array}\right](j=1,\ldots,q)$.\\
That is, $G(\cdot,j)$ is the $j$th column vector of matrix $G$.
Dually, $\langle\langle G|$ is a $1\times pq$ row vector defined as $(|G\rangle\rangle)^{\dagger}$, i.e., $\langle\langle G|=(|G\rangle\rangle)^{\dagger}$. (see \cite{Wood})
\end{definition}
\begin{remark}
\begin{enumerate}[(i)]
\item Vectorization is obviously linear: for matrices $S_{k}$ and scalars $\lambda_k$,\\ $|\sum_{k}\lambda_{k}S_{k}\rangle\rangle=\sum_{k}\lambda_{k}|S_{k}\rangle\rangle$.\\
\item Vectorization is inner-product-preserving; i.e. isometry. The Hilbert-Schmidt inner product is equivalent to the usual Euclidean inner product of vectors: for square matrices $S,T$ of the same size, $\langle S,T\rangle=\mbox{tr}(S^{\dagger}T)=\langle\langle S|T\rangle\rangle$. It is easily shown that vectorization is one-one and onto. Therefore vectorization is a unitary transformation from Hilbert-Schmidt matrix space to Hilbert vector space.\\
\item Vectorization is intrinsically related to the tensor product. Consider a square matrix of size $p\times p$, representing an operator acting on the $p$-dimensional Hilbert space $\mathcal{K}$. Let $\{|j\rangle\}_{j=1}^{p}$ be the orthonormal basis of $\mathcal{K}$ for which $|j\rangle$ is column vector with all entries 0 except for $j$th entry 1. A matrix $T=[t_{ij}]=\sum_{i,j=1}^{p}t_{ij}E_{ij}$, where $E_{ij}=|i\rangle\langle j|$, is transformed to the vector
\begin{eqnarray}
|T\rangle\rangle&=&|\sum_{i,j=1}^{p}t_{ij}E_{ij}\rangle\rangle=\sum_{i,j=1}^{p}t_{ij}|E_{ij}\rangle\rangle=\sum_{i,j=1}^{p}t_{ij}|i\rangle|j\rangle=\sum_{j=1}^{p}(\sum_{i=1}^{p}t_{ij}|i\rangle)|j\rangle\nonumber\\
             &=&\sum_{j=1}^{p}(T|j\rangle)|j\rangle=(T\otimes I_{p})(\sum_{i=1}^{p}|j\rangle|j\rangle)=(T\otimes I_{p})(\sum_{j=1}^{p}|E_{jj}\rangle\rangle)\nonumber\\
             &=&(T\otimes I_{p})|\sum_{j=1}^{p}E_{jj}\rangle\rangle=(T\otimes I_{p})|I_{p}\rangle\rangle=(I_{p}\otimes T^{\mbox{t}})|I_{p}\rangle\rangle.
\end{eqnarray}
Thus it follows from the identity above that, for any matrices $Q,X$ and $R$ of the same size $p\times p$,
\begin{eqnarray}
|QXR\rangle\rangle&=&(QXR)\otimes I_{p}|I_{p}\rangle\rangle=(Q\otimes I_{p})(X\otimes I_{p})[(R\otimes I_{p})|I_{p}\rangle\rangle]\nonumber\\
&=&(Q\otimes I_{p})(X\otimes I_{p})[(I_{p}\otimes R^{\mbox{t}})|I_{p}\rangle\rangle]
=(Q\otimes I_{p})[(X\otimes I_{p})(I_{p}\otimes R^{\mbox{t}})]|I_{p}\rangle\rangle\nonumber\\
&=&(Q\otimes I_{p})[(I_{p}\otimes R^{\mbox{t}})(X\otimes I_{p})]|I_{p}\rangle\rangle=(Q\otimes I_{p})(I_{p}\otimes R^{\mbox{t}})|X\rangle\rangle\nonumber\\
&=&Q\otimes R^{\mbox{t}}|X\rangle\rangle
\end{eqnarray}
and
\begin{eqnarray}
|XY\rangle\rangle=(X\otimes I_{p})|Y\rangle\rangle=(I_{p}\otimes Y^{\mbox{t}})|X\rangle\rangle.
\end{eqnarray}
\item For any matrix $Y$,
\begin{eqnarray}
\langle\langle Y^{\ast}|=(|Y^{\ast}\rangle\rangle)^{\dagger}=(|Y\rangle\rangle)^{\ast\dagger}=(|Y\rangle\rangle)^{\mbox{t}}.
\end{eqnarray}
\item For $S\in\mathscr{B}(\mathcal{H}_{A})$ and $T\in\mathscr{B}(\mathcal{H}_{B})$, where $\mathcal{H}_{A}=\mathcal{H}_{B}$ are $d$-dimensional Hilbert spaces. For the matrix representations $S=[s_{ij}]$ and $T=[t_{ij}](i,j=1,\ldots,d)$, we have
$\mbox{tr}_{B}(|S\rangle\rangle\langle\langle T|)=ST^{\dagger}$ and $\mbox{tr}_{A}(|S\rangle\rangle\langle\langle T|)=S^{\mbox{t}}T^{\ast}$. Indeed,
\begin{eqnarray*}
\mbox{tr}_{B}(|S\rangle\rangle\langle\langle T|)&=&\sum_{m,n,\mu,\nu=1}^{d}s_{mn}t^{\ast}_{\mu\nu}\mbox{tr}_{B}(|mn\rangle\langle\mu\nu|)=\sum_{m,n,\mu,\nu=1}^{d}s_{mn}t^{\ast}_{\mu\nu}\mbox{tr}_{B}(|m\rangle\langle\mu|\otimes |n\rangle\langle\nu|)\\
&=&\sum_{m,n,\mu,\nu=1}^{d}s_{mn}t^{\ast}_{\mu\nu}\delta_{n\nu}|m\rangle\langle\mu|=\sum_{m,n,\mu=1}^{d}s_{mn}t^{\ast}_{\mu n}|m\rangle\langle\mu|\\&=&\sum_{n=1}^{d}(\sum_{m=1}^{d}s_{mn}|m\rangle)(\sum_{\mu=1}^{d}t_{\mu n}|\mu\rangle)^{\dagger}=\sum_{n=1}^{d}(S|n\rangle)(T|n\rangle)^{\dagger}\\&=&\sum_{n=1}^{d}S|n\rangle\langle n|T^{\dagger}=ST^{\dagger}.
\end{eqnarray*}
\end{enumerate}
The other identity goes similarly.
\end{remark}
\begin{definition} Let $Z$ be an $d_{B}\times d_{B}$ block matrix with each entry of size $d_{A}\times d_{A}$; i.e. $Z=[Z_{\mu\nu}]$ represent an operator acting on $\mathcal{H}_{A}\otimes\mathcal{H}_{B}$. We define a realigned matrix $\mathcal{R}(Z)$, acting from $\mathcal{H}_{B}\otimes\mathcal{H}_{B}$ to $\mathcal{H}_{A}\otimes\mathcal{H}_{A}$, of size $d_{A}^{2}\times d_{B}^{2}$ that contains the same elements as $Z$ but in different position as\\
\indent
$\mathcal{R}(Z)=[|Z_{11}\rangle\rangle,\ldots,|Z_{d_{B}1}\rangle\rangle;\ldots;|Z_{1d_{B}}\rangle\rangle,\ldots,|Z_{d_{B}d_{B}}\rangle\rangle]$.\\
In fact, $\mathcal{R}(Z)_{\stackrel{mn}{\mu\nu}}=Z_{\stackrel{m\mu}{n\nu}}$. Similarly, we can also define another alignment $\mathcal{R}'$ as $\mathcal{R}'(Z)_{\stackrel{mn}{\mu\nu}}=Z_{\stackrel{\nu n}{\mu m}}$. Note that alignment of matrices is a one-one linear mapping from the matrix space $\mathcal{M}_{d_{A}d_{B}\times d_{A}d_{B}}(\mathbb{C})$ onto the matrix space $\mathcal{M}_{d_{A}^{2}\times d_{B}^{2}}(\mathbb{C})$.
\end{definition}
\begin{prop} \textit{For a tensor matrix $X\otimes Y$
with the factor matrix $X$ of size $d_{A}\times d_{A}$ and the
factor matrix $Y=[y_{\mu\nu}]$ of size $d_{B}\times d_{B}$,
$Z=[y_{\mu\nu}X]=[Z_{\mu\nu}]$. We have:
\begin{eqnarray}
\mathcal{R}(X\otimes
Y)=|X\rangle\rangle\langle\langle Y^{\ast}|.
\end{eqnarray}
Moreover, a nonzero
matrix $Z$ can be factorized as $X\otimes Y$ if and only if
$rank[\mathcal{R}(Z)]=1$.}
\end{prop}
\begin{proof}
\begin{eqnarray*}
\mathcal{R}(X\otimes Y)
              &=&[|y_{11}X\rangle\rangle,\ldots,|y_{d_{B}1}X\rangle\rangle;\ldots;|y_{1d_{B}}X\rangle\rangle,\ldots,|y_{d_{B}d_{B}}X\rangle\rangle]\\
              &=&[y_{11}|X\rangle\rangle,\ldots,y_{d_{B}1}|X\rangle\rangle;\ldots;y_{1d_{B}}|X\rangle\rangle,\ldots,y_{d_{B}d_{B}}|X\rangle\rangle]\\
              &=&|X\rangle\rangle[y_{11},\ldots,y_{d_{B}1};\ldots;y_{1d_{B}},\ldots,y_{d_{B}d_{B}}]=|X\rangle\rangle(|Y\rangle\rangle)^{\mbox{t}}\\
              &=&|X\rangle\rangle\langle\langle Y^{\ast}|.
\end{eqnarray*}
\end{proof}
For a general block matrix $Z$, it holds that
\begin{eqnarray}
\mathcal{R}(Z)&=&\mathcal{R}(\sum_{\mu,\nu=1}^{d_{B}}Z_{\mu\nu}\otimes|\mu\rangle\langle \nu|)=\sum_{\mu,\nu=1}^{d_{B}}\mathcal{R}(Z_{\mu\nu}\otimes|\mu\rangle\langle \nu|)\nonumber\\
&=&\sum_{\mu,\nu=1}^{d_{B}}|Z_{\mu\nu}\rangle\rangle(|\mu\nu\rangle)^{\mbox{t}}=\sum_{\mu,\nu=1}^{d_{B}}|Z_{\mu\nu}\rangle\rangle\langle \mu\nu|.
\end{eqnarray}
Before the properties of realignment derived, we need to know one useful operator called \emph{swap operator}, defined as $S=\sum_{i,j=1}^{N^{2}}|ij\rangle\langle ji|$, acting on $\mathcal{H}_{N}\otimes\mathcal{H}_{N}$.
Then by simple computations, we have:
\begin{prop}
For any $X$ and $Y$ of the same size $N\times N$. We have:
\begin{enumerate}[(i)]
\item $S$ is self-adjoint, unitary, symmetric, and orthogonal;\\
\item $|X^{\mbox{t}}\rangle\rangle=S|X\rangle\rangle$, $L_{\mathcal{T}}=S$;\\
\item $S(X\otimes Y)S=Y\otimes X$.
\end{enumerate}
\end{prop}
\begin{definition} With $S$ as above, the \emph{flip} transformation of matrices over a bipartite quantum system is defined as\\
\indent
$\mathcal{F}(Z)=SZS \ \mbox{with}\ \mathcal{F}(Z)_{\stackrel{m\mu}{n\nu}}=Z_{\stackrel{\mu m}{\nu n}}$.\\
Similarly, we can define two \emph{partial flips} as $\mathcal{F}_{r}(Z)=SZ\ \mbox{with}\ \mathcal{F}_{r}(Z)_{\stackrel{m\mu}{n\nu}}=Z_{\stackrel{\mu m}{n\nu}}$
and $\mathcal{F}_{c}(Z)=ZS\ \mbox{with}\ \mathcal{F}_{c}(Z)_{\stackrel{m\mu}{n\nu}}=Z_{\stackrel{m\mu}{\nu n}}$ (where `r' and `c' mean that row and column, respectively). Later, we will see that $L_{\mathcal{F}}=S\otimes S$.
\end{definition}
\begin{lem}\label{lem2.7}(\cite{Havel}) \textit{Given any two square matrices $X,Y$ of the same size, we have the following equation:
\begin{eqnarray}|X\otimes Y\rangle\rangle=(I\otimes S\otimes I)|X\rangle\rangle|Y\rangle\rangle.\end{eqnarray}}
\end{lem}
\begin{prop}
\begin{enumerate}[(i)]
 \item If $X,Y$ are matrices of the same size $N\times N$, then
\begin{eqnarray}|\mathcal{R}(X\otimes Y)\rangle\rangle=|X\rangle\rangle|Y\rangle\rangle;\end{eqnarray} i.e., the vectorization of the  matrix $|X\rangle\rangle\langle\langle Y^{\ast}|$ is $|X\rangle\rangle|Y\rangle\rangle$.\\
 \item Let $Z$ be a matrix of size $N^{2}\times N^{2}$. Then: $|\mathcal{R}(Z)\rangle\rangle=I\otimes S\otimes I|Z\rangle\rangle$, thus $L_{\mathcal{R}}=I\otimes S\otimes I$.\\
\item If $\Omega(\cdot)=\sum_{i,j=1}^{N}(I\otimes |i\rangle\langle
j|)\cdot(|i\rangle\langle j|\otimes I)$, then: for any matrices
$X,Y$ of the same size $N\times N$,
\begin{eqnarray}\Omega(|X\rangle\rangle\langle\langle Y|)=X\otimes Y^{\ast}\
\mbox{and}\  \Omega(X\otimes
Y^{\ast})=|X\rangle\rangle\langle\langle Y|=\mathcal{R}(X\otimes
Y^{\ast}).\end{eqnarray}
More generally, we have
$\Omega(Z)=\mathcal{R}(Z)$ for any matrix $Z$ of size $N^{2}\times N^{2}$.
\end{enumerate}
\end{prop}
\begin{proof} (i) and (ii) follow easily from Lemma \ref{lem2.7}.\\
(iii) Together with Lemma \ref{lem2.7}, it follows from (i) that
\begin{eqnarray*}
& &|\Omega(X\otimes Y^{\ast})\rangle\rangle\\
&=&|\sum_{i,j=1}^{N}(I\otimes |i\rangle\langle j|)X\otimes Y^{\ast}(|i\rangle\langle j|\otimes I)\rangle\rangle=\sum_{i,j=1}^{N}|(I\otimes |i\rangle\langle j|)X\otimes Y^{\ast}(|i\rangle\langle j|\otimes I)\rangle\rangle\\
&=&\sum_{i,j=1}^{N}(I\otimes |i\rangle\langle j|)\otimes(|j\rangle\langle i|\otimes I)|X\otimes Y^{\ast}\rangle\rangle
=\sum_{i,j=1}^{N}(I\otimes |ij\rangle\langle ji|\otimes I)|X\otimes Y^{\ast}\rangle\rangle\\
&=&(I\otimes S\otimes I)|X\otimes Y^{\ast}\rangle\rangle=|X\rangle\rangle|Y^{\ast}\rangle\rangle=|\mathcal{R}(X\otimes Y^{\ast})\rangle\rangle.
\end{eqnarray*}
Hence $\Omega(X\otimes Y^{\ast})=\mathcal{R}(X\otimes Y^{\ast})=|X\rangle\rangle\langle\langle Y|$. By simple computations, we have also $\Omega(|X\rangle\rangle\langle\langle Y|)=X\otimes Y^{\ast}$. Since
$\langle m\mu|\mathcal{R}(Z)|n\nu\rangle=\mathcal{R}(Z)_{\stackrel{m\mu}{n\nu}}=Z_{\stackrel{mn}{\mu\nu}}$ and
\\
\indent
$\langle m\mu|\sum_{i,j=1}^{N}(I\otimes |i\rangle\langle j|)Z(|i\rangle\langle j|\otimes I)|n\nu\rangle=\sum_{i,j=1}^{N}\delta_{\mu i}\delta_{nj}\langle mj|Z|i\nu\rangle=\langle mn|Z|\mu\nu\rangle=Z_{\stackrel{mn}{\mu\nu}}
$,\\
i.e., $\Omega(Z)=\mathcal{R}(Z)$. In such a way, we obtain the explicit expression for the realignment transformation:
\begin{eqnarray}
\mathcal{R}(Z)=\sum_{i,j=1}^{N}(I\otimes |i\rangle\langle j|)Z(|i\rangle\langle j|\otimes I)
\end{eqnarray}
for any matrix $Z$ of size $N^{2}\times N^{2}$.
\end{proof}
Next the relationship among \emph{the realignment, the
transposition, and the flip} over a bipartite quantum system will be
discussed.
First recall that the transposition $\mathcal{T}$ over bipartite quantum system $\mathcal{H}_{A}\otimes\mathcal{H}_{B}$ are defined as $\mathcal{T}(Z)\equiv\mathcal{T}_{A}\otimes\mathcal{T}_{B}(Z)$ with $\mathcal{T}(Z)_{\stackrel{m\mu}{n\nu}}=Z_{\stackrel{n\nu}{m\mu}}$, where $\mathcal{T}_{A}$ and $\mathcal{T}_{B}$ are the transpositions with respect to subsystems A and B, respectively. Apparently, $\mathcal{T}_{A}(Z)_{\stackrel{m\mu}{n\nu}}=Z_{\stackrel{n\mu}{m\nu}}$ and $\mathcal{T}_{B}(Z)_{\stackrel{m\mu}{n\nu}}=Z_{\stackrel{m\nu}{n\mu}}$.
\begin{prop}
\begin{enumerate}[(i)]
 \item $\mathcal{T},\mathcal{R}$ and $\mathcal{F}$ all are involution; i.e., $\mathcal{T}\circ\mathcal{T}=\mathcal{R}\circ\mathcal{R}=\mathcal{F}\circ\mathcal{F}=\mbox{Id}$.
 \item $\mathcal{F}\circ\mathcal{T}=\mathcal{T}\circ\mathcal{F}$, $\mathcal{T}\circ\mathcal{R}\neq\mathcal{R}\circ\mathcal{T}$ and $\mathcal{F}\circ\mathcal{R}\neq\mathcal{R}\circ\mathcal{F}$, where $\circ$ stands for the composite of transformations.
 \item $\mathcal{T}\circ\mathcal{R}=\mathcal{R}\circ\mathcal{F}$ and $\mathcal{R}\circ\mathcal{T}=\mathcal{F}\circ\mathcal{R}$.
 \item $\mathcal{R}'=\mathcal{T}\circ\mathcal{R}\circ\mathcal{T}=\mathcal{F}\circ\mathcal{R}\circ\mathcal{F}$.
 \item $\mathcal{F}_{r}=\mathcal{R}\circ\mathcal{T}_{A}\circ\mathcal{R}$ and $\mathcal{F}_{c}=\mathcal{R}\circ\mathcal{T}_{B}\circ\mathcal{R}$.
 \end{enumerate}
\end{prop}
\begin{proof} It is trivially by some computations. For example, $[\mathcal{T}\circ\mathcal{R}(X)]_{\stackrel{m\mu}{n\nu}}=
[\mathcal{T}(X)]_{\stackrel{mn}{\mu\nu}}=X_{\stackrel{\mu\nu}{mn}}$ and $[\mathcal{R}\circ\mathcal{F}(X)]_{\stackrel{m\mu}{n\nu}}=
[\mathcal{R}(X)]_{\stackrel{\mu m}{\nu n}}=X_{\stackrel{\mu\nu}{mn}}$; i.e.,
$[\mathcal{T}\circ\mathcal{R}(X)]_{\stackrel{m\mu}{n\nu}}=[\mathcal{R}\circ\mathcal{F}(X)]_{\stackrel{m\mu}{n\nu}}$
 which means that $\mathcal{T}\circ\mathcal{R}=\mathcal{R}\circ\mathcal{F}$. Others go similarly.
\end{proof}
\section{Dynamical matrices for quantum operations}
A density matrix
\begin{eqnarray}
\rho=\left[
                         \begin{array}{ccc}
                           \rho_{11} & \cdots & \rho_{1d_{B}} \\
                           \vdots & \vdots & \vdots \\
                           \rho_{d_{B}1} & \cdots & \rho_{d_{B}d_{B}} \\
                         \end{array}
                       \right]=[\rho_{\mu\nu}]
\end{eqnarray}
of size $d_{B}\times d_{B}$ may be treated as a vector
\begin{eqnarray}
|\rho\rangle\rangle=\left[
                                                               \begin{array}{c}
                                                                 \rho(\cdot,1) \\
                                                                 \vdots \\
                                                                 \rho(\cdot,d_{B}) \\
                                                               \end{array}
                                                             \right]
, \mbox{where}\ \rho(\cdot,\nu)=\left[
         \begin{array}{c}
           \rho_{1\nu} \\
           \vdots \\
           \rho_{d_{B}\nu} \\
         \end{array}
       \right](\nu=1,\ldots,d_{B}).
\end{eqnarray}
Suppose that $\rho$ and $\sigma$ act on $\mathcal{H}_{B}$ and
$\mathcal{H}_{A}$, respectively. The action of a linear
super-operator $\Phi:\rho\rightarrow
\sigma=\Phi(\rho)=[\sigma_{mn}]$ may thus be represented by a matrix
$L_{\Phi}\equiv L$ of size $d_{A}^2\times d_{B}^2$:
\begin{eqnarray}
|\sigma\rangle\rangle=|\Phi(\rho)\rangle\rangle=L|\rho\rangle\rangle\ \mbox{or}\ \sigma_{mn}=\sum_{\mu,\nu=1}^{d_{B}}L_{\stackrel{mn}{\mu\nu}}\rho_{\mu\nu}.
\end{eqnarray}
It can be written concretely as the equation of multiplicity of a \emph{supermatrix} and a \emph{supervector}:
\begin{eqnarray}
\left[\begin{array}{c}
        \sigma(\cdot,1) \\
        \vdots \\
        \sigma(\cdot,n) \\
        \vdots \\
        \sigma(\cdot, d_{A})
      \end{array}
\right]=\left[\begin{array}{ccccc}
                L_{11} & \cdots & L_{1\nu} & \cdots & L_{1d_{B}} \\
                \vdots & \vdots & \vdots & \vdots & \vdots \\
                L_{n1} & \cdots & L_{n\nu} & \cdots & L_{nd_{B}} \\
                \vdots & \vdots & \vdots & \vdots & \vdots \\
                 L_{d_{A}1} & \cdots & L_{d_{A}\nu} & \cdots & L_{d_{A}d_{B}}
              \end{array}
\right]\left[\begin{array}{c}
        \rho(\cdot,1) \\
        \vdots \\
        \rho(\cdot,\nu) \\
        \vdots \\
        \rho(\cdot, d_{B})
      \end{array}
\right],
\end{eqnarray}
where
\begin{eqnarray}
L_{n\nu}=[L_{\stackrel{mn}{\mu\nu}}]=\left[\begin{array}{ccc}
                                                     L_{\stackrel{1n}{1\nu}} & \cdots & L_{\stackrel{1n}{d_{B}\nu}} \\
                                                     \vdots & \vdots & \vdots \\
                                                     L_{\stackrel{d_{A}n}{1\nu}} & \cdots & L_{\stackrel{d_{A}n}{d_{B}\nu}}
                                                   \end{array}
\right].
\end{eqnarray}
One must be caution here that $n$ and $\nu$ stand for the block row index and the block column index, respectively; while $m$ and $\mu$ stand for row index and column index of each block. Now we give a simple example for a qubit map for later use as follows:
\begin{eqnarray}
\left[\begin{array}{c}
   \sigma_{11} \\
   \sigma_{21} \\
   \sigma_{12} \\
   \sigma_{22}
 \end{array}
\right]=\left[\begin{array}{cccc}
                L_{\stackrel{11}{11}} & L_{\stackrel{11}{21}} & L_{\stackrel{11}{12}} & L_{\stackrel{11}{22}} \\
                L_{\stackrel{21}{11}} & L_{\stackrel{21}{21}} & L_{\stackrel{21}{12}} & L_{\stackrel{21}{22}} \\
                L_{\stackrel{12}{11}} & L_{\stackrel{12}{21}} & L_{\stackrel{12}{12}} & L_{\stackrel{12}{22}} \\
                L_{\stackrel{22}{11}} & L_{\stackrel{22}{21}} & L_{\stackrel{22}{12}} & L_{\stackrel{22}{22}}
              \end{array}
\right]\left[\begin{array}{c}
   \rho_{11} \\
   \rho_{21} \\
   \rho_{12} \\
   \rho_{22}
 \end{array}
\right].
\end{eqnarray}
\begin{thrm} The requirement that the image $\sigma$ is a
density matrix, so it is Hermitian, positive with unit trace, impose
constraints on the matrix $L$:
\begin{enumerate}[(i)]
 \item $\sigma^{\dagger}=\sigma\ \Longrightarrow L^{\ast}_{\stackrel{mn}{\mu\nu}}=L_{\stackrel{nm}{\nu\mu}}$.
 \item $\sigma\geq 0\ \Longrightarrow\ [\sum_{\mu,\nu=1}^{d_{B}}L_{\stackrel{mn}{\mu\nu}}\rho_{\mu\nu}]\geq0$ for any state $\rho=[\rho_{\mu\nu}]$.
\item $\mbox{tr}(\sigma)=1\ \Longrightarrow\ \sum_{m=1}^{d_{A}}L_{\stackrel{mm}{\mu\nu}}=\delta_{\mu\nu}$.
\end{enumerate}
\end{thrm}
\begin{proof} (i)
\textbf{Step 1:} For state $\rho=|\gamma\rangle\langle\gamma|(\gamma\in\{1,\ldots,d_{B}\})$,
 \begin{eqnarray}
 \rho_{\mu\nu}=\langle \mu|\rho|\nu\rangle=\langle\mu|\gamma\rangle\langle\gamma|\nu\rangle=\delta_{\mu\gamma}\delta_{\nu\gamma}.
 \end{eqnarray}
 Then
\begin{eqnarray}
\sigma_{mn}=\sum_{\mu,\nu=1}^{d_{B}}L_{\stackrel{mn}{\mu\nu}}\delta_{\mu\gamma}\delta_{\nu\gamma}=L_{\stackrel{mn}{\gamma\gamma}}.
\end{eqnarray}
Since $\ \sigma^{\dagger}=\sigma$, it implies that $\sigma_{mn}=\sigma^{\ast}_{nm}$; i.e., $L_{\stackrel{mn}{\gamma\gamma}}=L^{\ast}_{\stackrel{nm}{\gamma\gamma}}(\gamma\in\{1,\ldots,d_{B}\})$.\\
\textbf{Step 2:} Setting
\begin{eqnarray}
\rho=\frac{1}{2}[|\alpha\rangle\langle\alpha|+|\beta\rangle\langle\beta|+|\alpha\rangle\langle
\beta|+|\beta\rangle\langle\alpha|]
(\alpha,\beta=1,\ldots,d_{B};\alpha\neq\beta),\label{state1}
\end{eqnarray}
we have
\begin{eqnarray}
\rho_{\mu\nu}=\frac{1}{2}[\delta_{\mu\alpha}\delta_{\nu\alpha}+
\delta_{\mu\beta}\delta_{\nu\beta}+\delta_{\mu\alpha}\delta_{\nu\beta}+\delta_{\mu\beta}\delta_{\nu\alpha}].\label{state2}
\end{eqnarray}
Hence
\begin{eqnarray}
\sigma_{mn}=\frac{1}{2}[L_{\stackrel{mn}{\alpha\alpha}}+L_{\stackrel{mn}{\beta\beta}}+L_{\stackrel{mn}{\alpha\beta}}+L_{\stackrel{mn}{\beta\alpha}}].
\end{eqnarray}

From the equation $\sigma_{mn}=\sigma^{\ast}_{nm}$, we know
$$L_{\stackrel{mn}{\alpha\alpha}}+L_{\stackrel{mn}{\beta\beta}}+L_{\stackrel{mn}{\alpha\beta}}+L_{\stackrel{mn}{\beta\alpha}}
=L^{\ast}_{\stackrel{nm}{\alpha\alpha}}+L^{\ast}_{\stackrel{nm}{\beta\beta}}+L^{\ast}_{\stackrel{nm}{\alpha\beta}}+L^{\ast}_{\stackrel{nm}{\beta\alpha}};$$
i.e.,
\begin{eqnarray}
L_{\stackrel{mn}{\alpha\beta}}+L_{\stackrel{mn}{\beta\alpha}}=L^{\ast}_{\stackrel{nm}{\alpha\beta}}+L^{\ast}_{\stackrel{nm}{\beta\alpha}}.\label{I}
\end{eqnarray}
\textbf{Step 3:} Letting
\begin{eqnarray}
\rho=\frac{1}{2}[|\alpha\rangle\langle\alpha|+|\beta\rangle\langle\beta|+\sqrt{-1}|\alpha\rangle\langle
\beta|-\sqrt{-1}|\beta\rangle\langle\alpha|],\label{state3}
\end{eqnarray}
we have
\begin{eqnarray}
\rho_{\mu\nu}=\frac{1}{2}[\delta_{\mu\alpha}\delta_{\nu\alpha}+\delta_{\mu\beta}\delta_{\nu\beta}+\sqrt{-1}\delta_{\mu\alpha}\delta_{\nu\beta}-\sqrt{-1}\delta_{\mu\beta}\delta_{\nu\alpha}].\label{state4}
\end{eqnarray}
Hence
\begin{eqnarray}\sigma_{mn}=\frac{1}{2}[L_{\stackrel{mn}{\alpha\alpha}}+L_{\stackrel{mn}{\beta\beta}}+\sqrt{-1}L_{\stackrel{mn}{\alpha\beta}}-\sqrt{-1}L_{\stackrel{mn}{\beta\alpha}}],
\end{eqnarray}
which implies that
\begin{eqnarray}\sigma^{\ast}_{nm}=\frac{1}{2}[L^{\ast}_{\stackrel{nm}{\alpha\alpha}}+L^{\ast}_{\stackrel{nm}{\beta\beta}}-\sqrt{-1}L^{\ast}_{\stackrel{nm}{\alpha\beta}}+\sqrt{-1}L^{\ast}_{\stackrel{nm}{\beta\alpha}}].
\end{eqnarray}
This gives rise to:
\begin{eqnarray}
L_{\stackrel{mn}{\alpha\beta}}-L_{\stackrel{mn}{\beta\alpha}}=-L^{\ast}_{\stackrel{nm}{\alpha\beta}}+L^{\ast}_{\stackrel{nm}{\beta\alpha}}\label{II}
\end{eqnarray}
Combing (\ref{I}) with (\ref{II}) gives that $L_{\stackrel{mn}{\alpha\beta}}=L^{\ast}_{\stackrel{nm}{\beta\alpha}}$.\\
(ii) is trivial.\\
(iii) Because $\mbox{tr}(\sigma)=1$, that is,
\begin{eqnarray}
1=\sum_{m=1}^{d_{A}}\sigma_{mm}=\sum_{m=1}^{d_{A}}\sum_{\mu,\nu=1}^{d_{B}}L_{\stackrel{mm}{\mu\nu}}\rho_{\mu\nu}.\label{III}
\end{eqnarray}
\textbf{Step 1:} Given
$\rho=|\gamma\rangle\langle\gamma|(\gamma\in\{1,\ldots,d_{B}\})$. So
$\rho_{\mu\nu}=\langle
\mu|\rho|\nu\rangle=\delta_{\mu\gamma}\delta_{\nu\gamma}.$ From the
equation (\ref{III}), we have that
\begin{eqnarray}
1=\sum_{m=1}^{d_{A}}\sum_{\mu,\nu=1}^{d_{B}}L_{\stackrel{mm}{\mu\nu}}\delta_{\mu\gamma}\delta_{\nu\gamma}=\sum_{m=1}^{d_{A}}L_{\stackrel{mm}{\gamma\gamma}}(\gamma\in\{1,\ldots,d_{B}\}).
\end{eqnarray}
\textbf{Step 2:} From the equation (\ref{state1}), (\ref{state2}) and
(\ref{III}), we have that
$$1=\frac{1}{2}\left[\sum_{m=1}^{d_{A}}L_{\stackrel{mm}{\alpha\alpha}}+\sum_{m=1}^{d_{A}}L_{\stackrel{mm}{\beta\beta}}+\sum_{m=1}^{d_{A}}L_{\stackrel{mm}{\alpha\beta}}+\sum_{m=1}^{d_{A}}L_{\stackrel{mm}{\beta\alpha}}\right];$$
i.e.,
\begin{eqnarray}
\sum_{m=1}^{d_{A}}L_{\stackrel{mm}{\alpha\beta}}+\sum_{m=1}^{d_{A}}L_{\stackrel{mm}{\beta\alpha}}=0.\label{A1}
\end{eqnarray}
\textbf{Step 3:} It follows from the equation (\ref{state3}) and
(\ref{state4}) that
\begin{eqnarray*}
1=\frac{1}{2}\left[\sum_{m=1}^{d_{A}}L_{\stackrel{mm}{\alpha\alpha}}+\sum_{m=1}^{d_{A}}L_{\stackrel{mm}{\beta\beta}}+\sqrt{-1}\sum_{m=1}^{d_{A}}L_{\stackrel{mm}{\alpha\beta}}-\sqrt{-1}\sum_{m=1}^{d_{A}}L_{\stackrel{mm}{\beta\alpha}}\right];
\end{eqnarray*}
i.e.,
\begin{eqnarray}
\sum_{m=1}^{d_{A}}L_{\stackrel{mm}{\alpha\beta}}-\sum_{m=1}^{d_{A}}L_{\stackrel{mm}{\beta\alpha}}=0.\label{A2}
\end{eqnarray}
From the equations (\ref{A1}) and (\ref{A2}), we get $\sum_{m=1}^{d_{A}}L_{\stackrel{mm}{\alpha\beta}}=0(\alpha\neq\beta)$. In summary, $\sum_{m=1}^{d_{A}}L_{\stackrel{mm}{\mu\nu}}=\delta_{\mu\nu}$.
\end{proof}
Note that the property (i) of the proposition 3.1. is not the
condition of Hermicity, and in general the matrix $L$ representing
the super-operator $\Phi$ is not Hermitian. However, by the
definition of matrix realignment we can define the
\emph{\textbf{dynamical matrix or Choi matrix}} (see \cite{Sudar,Karol}):
\\
\indent $D_{\Phi}\equiv \mathcal{R}(L)\ \mbox{with}\ D_{\stackrel{m\mu}{n\nu}}=L_{\stackrel{mn}{\mu\nu}}$.\\
In particular, the mapping $\mathcal{J}: \Phi\mapsto D_{\Phi}$ is called \emph{Choi-Jamio{\l}kowski isomorphism}.
\begin{prop} For a quantum channel $\Phi$, its dynamical matrix $D_{\Phi}$ enjoy the properties that follow:
\begin{enumerate}[(i)]
\item $D^{\dagger}_{\Phi}=D_{\Phi}$;
\item $D_{\Phi}\geq 0$;
\item $\mbox{tr}_{A}(D_{\Phi})=I_{B}$, $\mbox{Tr}(D_{\Phi})=N$;
 \item $|L_{\Phi}\rangle\rangle=(I\otimes S\otimes I)|D_{\Phi}\rangle\rangle; \langle\langle L_{\Phi}|L_{\Psi}\rangle\rangle=\langle\langle D_{\Phi}|D_{\Psi}\rangle\rangle; \langle L_{\Phi},L_{\Psi}\rangle=\langle D_{\Phi},D_{\Psi}\rangle$;
\item $\langle\Phi(X),Y\rangle=\langle D_{\Phi},Y\otimes X^{\ast}\rangle$ for any $X,Y$.
\end{enumerate}
\end{prop}
\begin{proof} Write $D_{\Phi}=D=[D_{\stackrel{m\mu}{n\nu}}]$.\\
(i) $D^{\dagger}=[D_{\stackrel{m\mu}{n\nu}}]^{\dagger}=[D^{\ast}_{\stackrel{m\mu}{n\nu}}]^{\mbox{t}}=[D^{\ast}_{\stackrel{n\nu}{m\mu}}]
=[L^{\ast}_{\stackrel{nm}{\nu\mu}}]=[L_{\stackrel{mn}{\mu\nu}}]=[D_{\stackrel{m\mu}{n\nu}}]=D$.\\
(ii) Let $|z\rangle=\sum_{n,\nu=1}^{N}z_{n\nu}|n\nu\rangle\rangle$. Then $\langle z|=\sum_{m,\mu=1}^{N}z^{\ast}_{m\mu}\langle\langle m\mu|$. Hence\\
\indent
$\langle z|D|z\rangle=\sum_{m,\mu,n,\nu=1}^{N}z^{\ast}_{m\mu}D_{\stackrel{m\mu}{n\nu}}z_{n\nu}$.\\
$|I\rangle\rangle=\sum_{m=1}^{N}|\mu\mu\rangle$ is called a maximally entangled state. So we have \\
\indent$|I\rangle\rangle\langle\langle I|=\sum_{\mu,\nu=1}^{N}|\mu\mu\rangle\langle \nu\nu|=\sum_{\mu,\nu=1}^{N}|\mu\rangle\langle\nu|\otimes|\mu\rangle\langle\nu|$.\\
Since $\Phi$ is completely positive map, $\Phi\otimes\mbox{Id}_{k}\geq 0(\forall \mbox{non-negative integer}\ k)$, in particular, $(\Phi\otimes\mbox{Id}_{N})(|I\rangle\rangle\langle\langle I|)\geq 0$, we get that
\begin{eqnarray*}
0&\leq&\langle z|(\Phi\otimes\mbox{Id}_{N})(|I\rangle\rangle\langle\langle I|)|z\rangle\\
 &=& \sum_{\mu,\nu=1}^{N}\langle z|[\Phi(|\mu\rangle\langle \nu|)\otimes|\mu\rangle\langle \nu|]|z\rangle\\
 &=& \sum_{\mu,\nu=1}^{N}\sum_{m,\alpha,n,\beta=1}^{N}z^{\ast}_{m\alpha}z_{n\beta}\langle m|\Phi(|\mu\rangle\langle \nu|)|n\rangle\cdot\langle\alpha|\mu\rangle\langle \nu|\beta\rangle\\
 &=& \sum_{m,\mu,n,\nu=1}z^{\ast}_{m\mu}z_{n\nu}L_{\stackrel{mn}{\mu\nu}}=\sum_{m,\mu,n,\nu=1}z^{\ast}_{m\mu}D_{\stackrel{m\mu}{n\nu}}z_{n\nu}.\Box
\end{eqnarray*}
Obviously, there is an identity in the proof: if
$\mathscr{D}(\mathcal{H}_{B})\stackrel{\Phi}{\longrightarrow}\mathscr{D}(\mathcal{H}_{A})$,
then
$\mathscr{D}(\mathcal{H}_{B}\otimes\mathcal{H}_{B})\stackrel{\Phi\otimes\mbox{\scriptsize
Id}_{N}}{\longrightarrow}\mathscr{D}(\mathcal{H}_{A}\otimes\mathcal{H}_{B})$,
\begin{eqnarray}
D_{\Phi}=(\Phi\otimes\mbox{Id}_{N})(|I\rangle\rangle\langle\langle
I|),
\Phi(\rho)=\mbox{Tr}_{A}[D_{\Phi}(I_{A}\otimes\rho^{\mbox{t}})].
\end{eqnarray}
Notes: If $X=|\mu\rangle\langle \nu|$, then
$|X\rangle\rangle=|\mu\rangle|\nu\rangle\equiv|\mu\nu\rangle$, from
which it follows that
$|\Phi(X)\rangle\rangle=L|X\rangle\rangle=L|\mu\nu\rangle=\sum_{i,j=1}^{N}|ij\rangle\langle
ij|
L|\mu\nu\rangle=\sum_{i,j=1}^{N}L_{\stackrel{ij}{\mu\nu}}|ij\rangle$.
Therefore,
$\Phi(X)=\sum_{i,j=1}^{N}L_{\stackrel{ij}{\mu\nu}}|i\rangle\langle
j|$, and
\begin{eqnarray*}
\langle m|\Phi(|\mu\rangle\langle \nu|)|n\rangle
    &=&\langle m|\Phi(X)|n\rangle=\sum_{i,j=1}^{N}L_{\stackrel{ij}{\mu\nu}}\langle m|i\rangle\langle j|n\rangle\\
    &=&\sum_{i,j=1}^{N}L_{\stackrel{ij}{\mu\nu}}\delta_{mi}\delta_{nj}=L_{\stackrel{mn}{\mu\nu}}.
\end{eqnarray*}
Since $\rho_{\mu\nu}=\langle\mu|\rho|\nu\rangle=\mbox{tr}(\rho|\nu\rangle\langle\mu|)=\mbox{tr}(|\mu\rangle\langle \nu|\rho^{\mbox{t}})$, we have:
\begin{eqnarray*}
\Phi(\rho)&=&\sum_{\mu,\nu}\rho_{\mu\nu}\Phi(|\mu\rangle\langle\nu|)=\sum_{\mu,\nu}\Phi(|\mu\rangle\langle \nu|)\mbox{tr}(|\mu\rangle\langle\nu|\rho^{\mbox{t}})=\sum_{\mu,\nu}\mbox{Tr}_{A}(\Phi(|\mu\rangle\langle \nu|)\otimes(|\mu\rangle\langle\nu|\rho^{\mbox{t}}))\\
&=&\mbox{Tr}_{A}[(\Phi\otimes\mbox{Id}_{N})(|I_{B}\rangle\rangle\langle\langle I_{B}|)(I_{A}\otimes\rho^{\mbox{t}})]
=\mbox{Tr}_{A}[D_{\Phi}(I_{A}\otimes\rho^{\mbox{t}})].
\end{eqnarray*}
(iii) Since $D=[D_{\mu\nu}]$, where $D_{\mu\nu}=[D_{\stackrel{m\mu}{n\nu}}]$, $\mbox{tr}_{A}D=[\mbox{tr}D_{\mu\nu}]$.
Because $\mbox{tr}D_{\mu\nu}=\sum_{m=1}^{N}D_{\stackrel{m\mu}{m\nu}}=\sum_{m=1}^{N}L_{\stackrel{mm}{\mu\nu}}=\delta_{\mu\nu}$,
thus we have $\mbox{tr}_{A}D=[\delta_{\mu\nu}]=I_{B}.$ Furthermore, $\mbox{Tr}(D)=N$ is trivially.\\
(iv) By the operator-sum representation theorem, we have $\Phi(\rho)=\sum_{j}\Gamma_{j}\rho\Gamma^{\dagger}_{j}$, thus
\begin{eqnarray*}
|L_{\Phi}\rangle\rangle&=&|\sum_{j}\Gamma_{j}\otimes\Gamma^{\ast}_{j}\rangle\rangle=\sum_{j}|\Gamma_{j}\otimes\Gamma^{\ast}_{j}\rangle\rangle
=\sum_{j}(I\otimes S\otimes I)|\Gamma_{j}\rangle\rangle|\Gamma^{\ast}_{j}\rangle\rangle\\
                       &=&\sum_{j}(I\otimes S\otimes I)||\Gamma_{j}\rangle\rangle\langle\langle\Gamma_{j}|\rangle\rangle=(I\otimes S\otimes I)|\sum_{j}\mathcal{R}(\Gamma_{j}\otimes\Gamma^{\ast}_{j})\rangle\rangle\\
&=&(I\otimes S\otimes I)|\mathcal{R}(\sum_{j}\Gamma_{j}\otimes\Gamma^{\ast}_{j})\rangle\rangle=(I\otimes S\otimes I)|\mathcal{R}(L)\rangle\rangle\\&=&(I\otimes S\otimes I)|D_{\Phi}\rangle\rangle.
\end{eqnarray*}
Therefore $\langle\langle L_{\Phi}|L_{\Psi}\rangle\rangle=
\langle\langle D_{\Phi}|(I\otimes S\otimes I)^{2}|D_{\Psi}\rangle\rangle=\langle\langle D_{\Phi}|D_{\Psi}\rangle\rangle$; that is $\langle L_{\Phi},L_{\Psi}\rangle=\langle D_{\Phi},D_{\Psi}\rangle$.\\
(v)\begin{eqnarray*}
\langle Y,\Phi(X)\rangle&=&\langle\langle Y|\Phi(X)\rangle\rangle=\langle\langle Y|L_{\Phi}|X\rangle\rangle=\mbox{Tr}[L_{\Phi}|X\rangle\rangle\langle\langle Y|]\\
&=&\mbox{Tr}[(|Y\rangle\rangle\langle\langle X|)^{\dagger}L_{\Phi}]=\langle |Y\rangle\rangle\langle\langle X|,L_{\Phi}\rangle=\langle\mathcal{R}(|Y\rangle\rangle\langle\langle X|),\mathcal{R}(L_{\Phi})\rangle\\
&=&\langle Y\otimes X^{\ast}, D_{\Phi}\rangle.
\end{eqnarray*}
\end{proof}
\noindent For any quantum channel $\Phi$, it induces its \emph{dual channel} $\Phi^{\dagger}$ in the following sense:\\
\indent$\langle\Phi(\rho),\sigma\rangle=\langle\rho,\Phi^{\dagger}(\sigma)\rangle\ \mbox{for any states}\ \rho\ \mbox{and}\ \sigma$.\\
If a CP map is given by the Kraus form $\Phi(\rho)=\sum_{j}\Gamma_{j}\rho\Gamma^{\dagger}_{j}$,
then the dual maps reads $\Phi^{\dagger}(\sigma)=\sum_{j}\Gamma^{\dagger}_{j}\rho\Gamma_{j}$.
Therefore, we have the following proposition into which the most useful results are summarized:
\begin{prop}
\begin{enumerate}[(i)]
\item $L_{\Phi}=\sum_{j}\Gamma_{j}\otimes\Gamma^{\ast}_{j}$, or $D_{\Phi}=\sum_{j}|\Gamma_{j}\rangle\rangle\langle\langle\Gamma_{j}|$ for $\Phi(\cdot)=\sum_{j}\Gamma_{j}\cdot\Gamma^{\dagger}_{j}$.
\item If $\Phi^{\dagger}$ is the dual channel of a quantum channel $\Phi$, then
$L_{\Phi^{\dagger}}=\mathcal{F}\circ\mathcal{T}(L_{\Phi})=L^{\dagger}_{\Phi}$, or $D_{\Phi^{\dagger}}=\mathcal{F}\circ\mathcal{T}(D_{\Phi})$.
\item $L_{r\Phi+s\Psi}=rL_{\Phi}+sL_{\Psi}$, or $D_{r\Phi+s\Psi}=rD_{\Phi}+sD_{\Psi}$.
\item the composition $\Phi\circ\Psi$ of two maps $\Phi$ and $\Psi$ means that $L_{\Phi\circ\Psi}=L_{\Phi}L_{\Psi}$, or $D_{\Phi\circ\Psi}=\mathcal{R}(\mathcal{R}(D_{\Phi})\mathcal{R}(D_{\Psi}))$.
\item $L_{\mathcal{T}\circ\Phi}=\mathcal{F}_{r}(L_{\Phi})$, or $D_{\mathcal{T}\circ\Phi}=\mathcal{T}_{A}(D_{\Phi})$; $L_{\Phi\circ\mathcal{T}}=\mathcal{F}_{c}(L_{\Phi})$, or $D_{\Phi\circ\mathcal{T}}=\mathcal{T}_{B}(D_{\Phi})$.
\item $L_{\mathcal{T}\circ\Phi\circ\mathcal{T}}=\mathcal{F}(L_{\Phi})=L^{\ast}_{\Phi}$, or $D_{\mathcal{T}\circ\Phi\circ\mathcal{T}}=\mathcal{T}(D_{\Phi})=D^{\mbox{t}}_{\Phi}=D^{\ast}_{\Phi}$.
\end{enumerate}
\end{prop}
\begin{proof} (i) Since $L_{\Phi}|X\rangle\rangle=|\Phi(X)\rangle\rangle=|\sum_{j}\Gamma_{j}X\Gamma^{\dagger}_{j}\rangle\rangle=\sum_{j}\Gamma_{j}\otimes\Gamma^{\ast}_{j}|X\rangle\rangle$,
 $L_{\Phi}=\sum_{j}\Gamma_{j}\otimes\Gamma^{\ast}_{j}$. $D_{\Phi}=\mathcal{R}(L_{\Phi})=\sum_{j}\mathcal{R}(\Gamma_{j}\otimes\Gamma^{\ast}_{j})=\sum_{j}|\Gamma_{j}\rangle\rangle\langle\langle\Gamma_{j}|$.\\
(ii) Obviously, $L_{\Phi^{\dagger}}=\sum_{j}\Gamma_{j}^{\dagger}\otimes\Gamma^{\mbox{t}}_{j}=(\sum_{j}\Gamma_{j}\otimes\Gamma^{\ast}_{j})^{\dagger}=L^{\dagger}_{\Phi}$. Thus it follows from (3) of Proposition 2.9. that
\begin{eqnarray*}
D_{\Phi^{\dagger}}&=&\mathcal{R}(L_{\Phi^{\dagger}})=\mathcal{R}(L^{\dagger}_{\Phi})=\mathcal{R}\circ\mathcal{T}(L^{\ast}_{\Phi})\\
&=&[\mathcal{R}\circ\mathcal{T}(L_{\Phi})]^{\ast}=[\mathcal{F}\circ\mathcal{R}(L_{\Phi})]^{\ast}=[\mathcal{F}(D_{\Phi})]^{\ast}=\mathcal{F}(D^{\ast}_{\Phi})=\mathcal{F}\circ\mathcal{T}(D_{\Phi}).
\end{eqnarray*}
(iii) It is trivially because
\begin{eqnarray*}
L_{r\Phi+s\Psi}|X\rangle\rangle&=&|(r\Phi+s\Psi)(X)\rangle\rangle=r|\Phi(X)\rangle\rangle+s|\Psi(X)\rangle\rangle
\\&=&rL_{\Phi}|X\rangle\rangle+sL_{\Psi}|X\rangle\rangle=(rL_{\Phi}+sL_{\Psi})|X\rangle\rangle.
\end{eqnarray*}
$D_{r\Phi+s\Psi}=rD_{\Phi}+sD_{\Psi}$ holds since the reshuffle transformation is linear.\\
(iv) $L_{\Phi\circ\Psi}|X\rangle\rangle=|\Phi\circ\Psi(X)\rangle\rangle=L_{\Phi}|\Psi(X)\rangle\rangle=L_{\Phi}L_{\Psi}|X\rangle\rangle$. This implies that $D_{\Phi\circ\Psi}=\mathcal{R}(\mathcal{R}(D_{\Phi})\mathcal{R}(D_{\Psi}))$.\\
(v) $L_{\mathcal{T}\circ\Phi}=L_{\mathcal{T}}L_{\Phi}=SL_{\Phi}=\mathcal{F}_{r}(L_{\Phi})$; similarly, $L_{\Phi\circ\mathcal{T}}=L_{\Phi}L_{\mathcal{T}}=L_{\Phi}S=\mathcal{F}_{c}(L_{\Phi})$. Thus\\
\indent
$
D_{\mathcal{T}\circ\Phi}=\mathcal{R}(L_{\mathcal{T}\circ\Phi})=\mathcal{R}\circ\mathcal{F}_{r}(L_{\Phi})=\mathcal{R}\circ\mathcal{F}_{r}\circ\mathcal{R}(D_{\Phi})=\mathcal{T}_{A}(D_{\Phi})
$\\
and\\
\indent
$D_{\Phi\circ\mathcal{T}}=\mathcal{R}(L_{\Phi\circ\mathcal{T}})=\mathcal{R}\circ\mathcal{F}_{c}(L_{\Phi})=\mathcal{R}\circ\mathcal{F}_{c}\circ\mathcal{R}(D_{\Phi})=\mathcal{T}_{B}(D_{\Phi}).
$\\
(vi) $L_{\mathcal{T}\circ\Phi\circ\mathcal{T}}=L_{\mathcal{T}}L_{\Phi}L_{\mathcal{T}}=SL_{\Phi}S=\mathcal{F}(L_{\Phi})$. Thus\\
\indent
$D_{\mathcal{T}\circ\Phi\circ\mathcal{T}}=\mathcal{R}(L_{\mathcal{T}\circ\Phi\circ\mathcal{T}})=\mathcal{R}\circ\mathcal{F}(L_{\Phi})
=\mathcal{R}\circ\mathcal{F}\circ\mathcal{R}(D_{\Phi})=\mathcal{T}(D_{\Phi})=D^{\mbox{t}}_{\Phi}=D^{\ast}_{\Phi}.
$
\end{proof}
\begin{prop} For two quantum operations $\Phi,\Psi$ on the $N$-dimensional identical  subsystems, $\mathcal{H}_{A},\mathcal{H}_{B}$ of a bipartite quantum system $\mathcal{H}_{A}\otimes\mathcal{H}_{B}$, respectively.
Then:
\begin{eqnarray}
L_{\Phi\otimes\Psi}=(I\otimes S\otimes I)(L_{\Phi}\otimes L_{\Psi})(I\otimes S\otimes I).
\end{eqnarray}
\end{prop}
\begin{proof} $\rho=[\rho_{\mu\nu}]=\sum_{\mu,\nu=1}^{N}\rho_{\mu\nu}\otimes |\mu\rangle\langle\nu|,\ \mbox{where}\ \rho_{\mu\nu}=[\rho_{\stackrel{m\mu}{n\nu}}]$,
is a $N\times N$ block density matrix whose entries being $N\times N$ scalar matrices. Since
\begin{eqnarray*}
(I\otimes S\otimes I)|\rho\rangle\rangle&=&(I\otimes S\otimes I)\sum_{\mu,\nu=1}^{N}|\rho_{\mu\nu}\otimes|\mu\rangle\langle\nu|\rangle\rangle\\
&=&\sum_{\mu,\nu=1}^{N}(I\otimes S\otimes I)^{2}|\rho_{\mu\nu}\rangle\rangle\otimes|\mu\nu\rangle\\
&=&\sum_{\mu,\nu=1}^{N}|\rho_{\mu\nu}\rangle\rangle\otimes|\mu\nu\rangle,
\end{eqnarray*}
we can get that
\begin{eqnarray*}
L_{\Phi\otimes\Psi}|\rho\rangle\rangle &=&|(\Phi\otimes\Psi)(\rho)\rangle\rangle=\sum_{\mu,\nu=1}^{N}|\Phi(\rho_{\mu\nu})\otimes\Psi(|\mu\rangle\langle\nu|)\rangle\rangle\\ &=&\sum_{\mu,\nu=1}^{N}(I\otimes S\otimes I)[|\Phi(\rho_{\mu\nu})\rangle\rangle\otimes|\Psi(|\mu\rangle\langle\nu|)\rangle\rangle]\\
&=&\sum_{\mu,\nu=1}^{N}(I\otimes S\otimes I)[L_{\Phi}|\rho_{\mu\nu}\rangle\rangle\otimes L_{\Psi}|\mu\nu\rangle]\\
&=&\sum_{\mu,\nu=1}^{N}(I\otimes S\otimes I)(L_{\Phi}\otimes L_{\Psi})[|\rho_{\mu\nu}\rangle\rangle\otimes |\mu\nu\rangle]\\
&=&(I\otimes S\otimes I)(L_{\Phi}\otimes L_{\Psi})[\sum_{\mu,\nu=1}^{N}|\rho_{\mu\nu}\rangle\rangle\otimes |\mu\nu\rangle]\\
&=&(I\otimes S\otimes I)(L_{\Phi}\otimes L_{\Psi})(I\otimes S\otimes I)|\rho\rangle\rangle.
\end{eqnarray*}
\end{proof}
\begin{prop} Let $\Phi,\Psi$ be two quantum operations on $\mathcal{H}_{N}$. If $\rho,\sigma$ are states in
$\mathcal{H}_{N}\otimes\mathcal{H}_{N}$ and $\sigma=(\Phi\otimes\Psi)(\rho)$, then:
\begin{eqnarray}\mathcal{R}(\sigma)=L_{\Phi}\mathcal{R}(\rho)L^{\mbox{t}}_{\Psi}.\end{eqnarray}
\end{prop}
\begin{proof}
\begin{eqnarray*}
|\mathcal{R}(\sigma)\rangle\rangle&=&(I\otimes S\otimes I)|\sigma\rangle\rangle=(I\otimes S\otimes I)|(\Phi\otimes\Psi)(\rho)\rangle\rangle=(I\otimes S\otimes I)L_{\Phi\otimes\Psi}|\rho\rangle\rangle\\
&=&(I\otimes S\otimes I)(I\otimes S\otimes I)(L_{\Phi}\otimes L_{\Psi})(I\otimes S\otimes I)|\rho\rangle\rangle=(L_{\Phi}\otimes L_{\Psi})|\mathcal{R}(\rho)\rangle\rangle\\
&=&|L_{\Phi}\mathcal{R}(\rho)L^{\mbox{t}}_{\Psi}\rangle\rangle.
\end{eqnarray*}
\end{proof}
\begin{lem} The composition of two completely positive linear super-operators $\Phi\circ\Psi$ is again completely positive.
\end{lem}
\begin{proof} To see that $\Phi\circ\Psi$ is completely positive,
it suffices to show $(\Phi\circ\Psi)\otimes\mbox{Id}_{k}$ is
positive for any $k\in\mathbb{N}$. Obviously,
$(\Phi\circ\Psi)\otimes\mbox{Id}_{k}=(\Phi\otimes\mbox{Id}_{k})\circ(\Psi\otimes\mbox{Id}_{k})$.
Since $\Phi$ and $\Psi$ are CP maps, $\Phi\otimes\mbox{Id}_{k}$ and
$\Psi\otimes\mbox{Id}_{k}$ are positive for any $k\in\mathbb{N}$,
which implies that the composition
$(\Phi\otimes\mbox{Id}_{k})\circ(\Psi\otimes\mbox{Id}_{k})$ is positive for any $k\in\mathbb{N}$.
\end{proof}
\begin{cor} Given two Hermitian matrices $A,B$
of the same size $N^{2}\times N^{2}$. If $A,B\geq0$,
then $\mathcal{R}(\mathcal{R}(A)\mathcal{R}(B))\geq0$.
\end{cor}
\begin{proof} We can consider two non-negative matrices $A$ and
$B$ of the same size $N^{2}\times N^{2}$ as the dynamical matrices
for two linear super-operators $\Phi_{A}$ and $\Phi_{B}$, both
acting from $\mathbb{M}_{N}$ to $\mathbb{M}_{N}$, respectively. Now
$A,B\geq0$ imply that $\Phi_{A}$ and $\Phi_{B}$ are CP maps. Thus
their composition $\Phi_{A}\circ\Phi_{B}$ is CP map by Lemma 3.6. It
follows from this that, for $\Phi_{A}\circ\Phi_{B}$, its dynamical
matrix
$D_{\Phi_{A}\circ\Phi_{B}}=\mathcal{R}(\mathcal{R}(A)\mathcal{R}(B))$
is non-negative. The result that follows immediately. Another complicated proof on the present corollary can be found in \cite{Havel}.
\end{proof}
\begin{cor} Given a finite set of Hermitian matrices $\{D_{j}: j=1,\ldots,n\}$ of the same size $N^{2}\times N^{2}$.
If $D_{j}\geq0$ for all $j$, then $\mathcal{R}(\mathcal{R}(D_{n})\mathcal{R}(D_{n-1})\cdots \mathcal{R}(D_{1}))\geq0$.
\end{cor}
\begin{proof} For each positive matrix $D_j$ of size $N^{2}\times N^{2}$, linear super-operator $\Phi_j$
determined by  $D_j$ is completely positive.  Thus the composition of $n$ completely positive
linear super-operators $\{D_{j}: j=1,\ldots,n\}$ is denoted by $\Phi=\Phi_{n}\circ\cdots\circ\Phi_1$. Therefore the dynamical matrix
$D_{\Phi}$ for $\Phi$ is equal to $\mathcal{R}(\mathcal{R}(D_{n})\mathcal{R}(D_{n-1})\cdots \mathcal{R}(D_{1}))$.
Since composition preserves completely positivity by the above lemma, $\Phi$ is completely positive,
therefore $D_{\Phi}\geq0$.
\end{proof}
\begin{prop} The Hilbert-Schmidt inner product, i.e., $\langle X,Y\rangle=\mbox{Tr}(X^{\dagger}Y)$,  on the matrix space
$\mathcal{M}_{N}$ induces another inner product in the space of linear maps $\mathscr{L}(\mathcal{M}_{N}, \mathcal{M}_{N})$.
\end{prop}
\begin{proof} Let $\{E_{\alpha}:\alpha=1,\ldots,N^{2}\}$ and $\{F_{\alpha}:\alpha=1,\ldots,N^{2}\}$ be orthonormal bases in $\mathcal{M}_{N}$, where $\langle E_{\alpha},E_{\beta}\rangle=\langle F_{\alpha},F_{\beta}\rangle=\delta_{\alpha\beta}$. We need only to prove that\\
\indent$\sum_{\alpha=1}^{N^{2}}\mbox{Tr}\Phi(E_{\alpha})^{\dagger}\Psi(E_{\alpha})=\sum_{\alpha=1}^{N^{2}}\mbox{Tr}\Phi(F_{\alpha})^{\dagger}\Psi(F_{\alpha})$.\\
Since $|E_{\alpha}\rangle\rangle=\sum_{\beta=1}^{N^{2}}|F_{\beta}\rangle\rangle\langle\langle F_{\beta}|E_{\alpha}\rangle\rangle=\sum_{\beta=1}^{N^{2}}\langle\langle F_{\beta}|E_{\alpha}\rangle\rangle|F_{\beta}\rangle\rangle$, $E_{\alpha}=\sum_{\beta=1}^{N^{2}}\langle\langle F_{\beta}|E_{\alpha}\rangle\rangle F_{\beta}$.
\begin{eqnarray*}
\sum_{\alpha=1}^{N^{2}}\mbox{Tr}\Phi(E_{\alpha})^{\dagger}\Psi(E_{\alpha})&=&\sum_{\alpha=1}^{N^{2}}\sum_{\beta,\gamma=1}^{N^{2}}\overline{\langle\langle F_{\beta}|E_{\alpha}\rangle\rangle}\langle\langle F_{\gamma}|E_{\alpha}\rangle\rangle\mbox{Tr}\Phi(F_{\beta})^{\dagger}\Psi(F_{\gamma})\\
&=&\sum_{\alpha=1}^{N^{2}}\sum_{\beta,\gamma=1}^{N^{2}}\langle\langle F_{\gamma}|E_{\alpha}\rangle\rangle\langle\langle E_{\alpha}|F_{\beta}\rangle\rangle\mbox{Tr}\Phi(F_{\beta})^{\dagger}\Psi(F_{\gamma})\\
&=&\sum_{\beta,\gamma=1}^{N^{2}}\langle\langle F_{\gamma}|(\sum_{\alpha=1}^{N^{2}}|E_{\alpha}\rangle\rangle\langle\langle E_{\alpha}|)|F_{\beta}\rangle\rangle\mbox{Tr}\Phi(F_{\beta})^{\dagger}\Psi(F_{\gamma})\\
&=&\sum_{\beta,\gamma=1}^{N^{2}}\langle\langle F_{\gamma}|F_{\beta}\rangle\rangle\mbox{Tr}\Phi(F_{\beta})^{\dagger}\Psi(F_{\gamma})=\sum_{\beta,\gamma=1}^{N^{2}}\delta_{\beta\gamma}\mbox{Tr}\Phi(F_{\beta})^{\dagger}\Psi(F_{\gamma})\\
&=&\sum_{\alpha=1}^{N^{2}}\mbox{Tr}\Phi(F_{\alpha})^{\dagger}\Psi(F_{\alpha}).
\end{eqnarray*}
\end{proof}
Now we define the inner product of two linear super-operators $\Phi$ and $\Psi$ (see \cite{Asorey}) as follows:\\
\begin{eqnarray}
\langle\Phi,\Psi\rangle\equiv\sum_{\alpha=1}^{N^{2}}\mbox{Tr}\Phi(E_{\alpha})^{\dagger}\Psi(E_{\alpha}).
\end{eqnarray}
Using this correspondence it is possible to introduce two different bases, associated to the bases $\{E_{\alpha}\}_{\alpha=1}^{N^{2}},\{F_{\beta}\}_{\beta=1}^{N^{2}}$,  in the space of linear maps:

\begin{enumerate}[1)]
\item \textbf{Type--I basis $\{\Delta_{\alpha\beta}\}$ in
$\mathscr{L}(\mathcal{M}_{N}, \mathcal{M}_{N})$} is defined by
\begin{eqnarray}
\Delta_{\alpha\beta}(X)=E_{\alpha}\langle
F_{\beta},X\rangle=E_{\alpha}\mbox{Tr}F^{\dagger}_{\beta}X,\
X\in\mathbb{M}_{N};
\end{eqnarray}
and
\item \textbf{Type--II basis $\{\Theta_{\alpha\beta}\}$ in $\mathscr{L}(\mathcal{M}_{N}, \mathcal{M}_{N})$} is defined by
\begin{eqnarray}
\Theta_{\alpha\beta}(X)=E_{\alpha}XF^{\dagger}_{\beta},\ X\in\mathbb{M}_{N}.
\end{eqnarray}
\end{enumerate}
Indeed, 1) Let
$\sum_{\alpha,\beta=1}^{N^{2}}c_{\alpha\beta}\Delta_{\alpha\beta}=0$
for some scalars $c_{\alpha\beta}\in\mathbb{C}$. This implies that
$\sum_{\alpha,\beta=1}^{N^{2}}c_{\alpha\beta}\Delta_{\alpha\beta}(X)=0$,
in particular, for $X=F_{\gamma}(\gamma=1,\ldots,N^{2})$, we have:\\
\indent$
0=\sum_{\alpha,\beta=1}^{N^{2}}c_{\alpha\beta}\Delta_{\alpha\beta}(F_{\gamma})=\sum_{\alpha,\beta=1}^{N^{2}}c_{\alpha\beta}\delta_{\beta\gamma}E_{\alpha}=\sum_{\alpha=1}^{N^{2}}c_{\alpha\gamma}E_{\alpha}
$\\
Since $\{E_{\alpha}\}$ is linearly independent,
$c_{\alpha\gamma}=0(\alpha,\gamma=1,\ldots,N^{2})$. We have also
that
$\langle\Delta_{\alpha\beta},\Delta_{\mu\nu}\rangle=\sum_{i,j,k,l=1}^{N^{2}}\mbox{Tr}[\Delta_{\alpha\beta}(|i\rangle\langle
j|)^{\dagger}\Delta_{\mu\nu}(|k\rangle\langle
l|)]=\delta_{\alpha\mu}\delta_{\beta\nu}$. Furthermore,
$L_{\Delta_{\alpha\beta}}=|E_{\alpha}\rangle\rangle\langle\langle
F_{\beta}|$.\\
2) Let
$\sum_{\alpha,\beta=1}^{N^{2}}c_{\alpha\beta}\Theta_{\alpha\beta}=0$
for some scalars $c_{\alpha\beta}\in\mathbb{C}$. This implies that
$\sum_{\alpha,\beta=1}^{N^{2}}c_{\alpha\beta}\Theta_{\alpha\beta}(X)=0$,
we have:\\
\indent$0=\sum_{\alpha,\beta=1}^{N^{2}}c_{\alpha\beta}\Theta_{\alpha\beta}(X)=\sum_{\alpha,\beta=1}^{N^{2}}c_{\alpha\beta}E_{\alpha}XF^{\dagger}_{\beta}$,\\
which means that
\begin{eqnarray*}
0&=&|\sum_{\alpha,\beta=1}^{N^{2}}c_{\alpha\beta}\Theta_{\alpha\beta}(X)\rangle\rangle=
\sum_{\alpha,\beta=1}^{N^{2}}c_{\alpha\beta}|\Theta_{\alpha\beta}(X)\rangle\rangle=\sum_{\alpha,\beta=1}^{N^{2}}c_{\alpha\beta}|E_{\alpha}XF^{\dagger}_{\beta}\rangle\rangle\\
 &=&\left(\sum_{\alpha,\beta=1}^{N^{2}}c_{\alpha\beta}E_{\alpha}\otimes F^{\ast}_{\beta}\right)|X\rangle\rangle;
\end{eqnarray*}
i.e., $\sum_{\alpha,\beta=1}^{N^{2}}c_{\alpha\beta}E_{\alpha}\otimes
F^{\ast}_{\beta}=0$ since $X$ is arbitrary. Because of the
independence of the set $\{E_{\alpha}\otimes
F^{\ast}_{\beta}\}^{N^{2}}_{\alpha,\beta=1}$, this implies that
$c_{\alpha\beta}=0(\alpha,\beta=1,\ldots,N^{2})$. And we have also
that
$\langle\Theta_{\alpha\beta},\Theta_{\mu\nu}\rangle=\delta_{\alpha\mu}\delta_{\beta\nu}$.
Furthermore, $L_{\Theta_{\alpha\beta}}=E_{\alpha}\otimes
F^{\ast}_{\beta}.\Box$
\begin{remark} Therefore, according to two kind of the
above-mentioned bases, we can expanding any mapping
$\Phi\in\mathscr{L}(\mathcal{M}_{N}, \mathcal{M}_{N})$ with respect to
Type--I and Type--II bases, respectively, to get two expressions
that follow:
\begin{eqnarray}
\Phi=\sum_{\alpha,\beta=1}^{N^{2}}p_{\alpha\beta}\Delta_{\alpha\beta}=\sum_{\alpha,\beta=1}^{N^{2}}q_{\alpha\beta}\Theta_{\alpha\beta}.
\end{eqnarray}
Now $L_{\Phi}=\sum_{\alpha,\beta=1}^{N^{2}}p_{\alpha\beta}|E_{\alpha}\rangle\rangle\langle\langle F_{\beta}|=\sum_{\alpha,\beta=1}^{N^{2}}q_{\alpha\beta}E_{\alpha}\otimes F^{\ast}_{\beta}$. We write $P=[p_{\alpha\beta}], Q=[q_{\alpha\beta}]$. \\
There is natural question to be asked: what is the relationships among these matrices $P, Q$? (see \cite{Nambu})
\end{remark}
\begin{prop} With the above notations,
\begin{eqnarray}
\langle\Delta_{\alpha\beta},\Theta_{\mu\nu}\rangle=\langle\Theta_{\alpha\beta},\Delta_{\mu\nu}\rangle=\mbox{Tr}(E^{\dagger}_{\alpha}E_{\mu}F_{\beta}F^{\dagger}_{\nu}).
\end{eqnarray}
Thus
\begin{enumerate}[(i)]
\item
$p_{\alpha\beta}=\sum_{\mu,\nu=1}^{N^{2}}\mbox{Tr}(E^{\dagger}_{\alpha}E_{\mu}F_{\beta}F^{\dagger}_{\nu})q_{\mu\nu}$;
\item
$q_{\alpha\beta}=\sum_{\mu,\nu=1}^{N^{2}}\mbox{Tr}(E^{\dagger}_{\alpha}E_{\mu}F_{\beta}F^{\dagger}_{\nu})p_{\mu\nu}$.
\end{enumerate}
\end{prop}
\begin{proof} By the definition of the inner product in the space
of linear maps,
\begin{eqnarray*}
\langle\Delta_{\alpha\beta},\Theta_{\mu\nu}\rangle&=&\sum_{i,j=1}^{N}\mbox{Tr}((\Delta_{\alpha\beta}(|i\rangle\langle
j|))^{\dagger}\Theta_{\mu\nu}(|i\rangle\langle
j|))=\sum_{i,j=1}^{N}\overline{\langle
j|F_{\beta}^{\dagger}|i\rangle}\mbox{Tr}(E^{\dagger}_{\alpha}E_{\mu}|i\rangle\langle
j|F^{\dagger}_{\nu})\\
&=&\sum_{i,j=1}^{N}\langle i|F_{\beta}|j\rangle\cdot\langle
j|F^{\dagger}_{\nu}
E^{\dagger}_{\alpha}E_{\mu}|i\rangle=\sum_{i=1}^{N}\langle
i|F_{\beta}F^{\dagger}_{\nu}
E^{\dagger}_{\alpha}E_{\mu}|i\rangle=\mbox{Tr}(F_{\beta}F^{\dagger}_{\nu}
E^{\dagger}_{\alpha}E_{\mu})\\
&=&\mbox{Tr}(E^{\dagger}_{\alpha}E_{\mu}F_{\beta}F^{\dagger}_{\nu}).
\end{eqnarray*}
Similarly, we have also:
$\langle\Theta_{\alpha\beta},\Delta_{\mu\nu}\rangle=\mbox{Tr}(E^{\dagger}_{\alpha}E_{\mu}F_{\beta}F^{\dagger}_{\nu})$.
Since
\begin{eqnarray*}
p_{\alpha\beta}&=&\langle\Delta_{\alpha\beta},\Phi\rangle=\sum_{\mu,\nu=1}^{N^{2}}
\langle\Delta_{\alpha\beta},\Theta_{\mu\nu}\rangle\langle\Theta_{\mu\nu},\Phi\rangle\\&
=&\sum_{\mu,\nu=1}^{N^{2}}\langle\Delta_{\alpha\beta},\Theta_{\mu\nu}\rangle
q_{\mu\nu},
\end{eqnarray*}
1) and 2) is trivial.
\end{proof}
\begin{remark}
\begin{enumerate}[(i)]
\item A special case is provided by the choice
$E_{\alpha}=F_{\alpha}$ (or $E_{\alpha}=F_{\alpha}=|i\rangle\langle
j|$, where $\{|i\rangle\}_{i=1}^{N}$ are an orthonormal basis for
$\mathbb{C}^{N}$) (see \cite{Pablo}).\\
 \item Since $|I\rangle\rangle=|\sum_{i}|i\rangle\langle i|\rangle\rangle=\sum_{i}|ii\rangle$, $|I\rangle\rangle\langle\langle I|=\sum_{i,j}|ii\rangle\rangle\langle\langle jj|=\sum_{i,j}|i\rangle\langle j|\otimes |i\rangle\langle j|$. We know that $I\otimes I=\sum_{\alpha=1}^{N^{2}}|E_{\alpha}\rangle\rangle\langle\langle E_{\alpha}|$ when $\{E_{\alpha}\}_{\alpha=1}^{N^{2}}$ is an orthonormal basis for $\mathcal{M}_{N}$. Thus we have:
$|I\rangle\rangle\langle\langle I|=\mathcal{R}(I\otimes I)=\sum_{\alpha=1}^{N^{2}}\mathcal{R}(|E_{\alpha}\rangle\rangle\langle\langle E_{\alpha}|)=\sum_{\alpha=1}^{N^{2}}E_{\alpha}\otimes E^{\ast}_{\alpha}$.
If there is another orthonormal basis $\{F_{\beta}\}_{\beta=1}^{N^{2}}$ for $\mathcal{M}_{N}$,
we still have: $|I\rangle\rangle\langle\langle I|=\sum_{\beta=1}^{N^{2}}F_{\beta}\otimes F^{\ast}_{\beta}$.
Therefore, $\sum_{\alpha=1}^{N^{2}}E_{\alpha}\otimes E^{\ast}_{\alpha}=\sum_{\beta=1}^{N^{2}}F_{\beta}\otimes F^{\ast}_{\beta}=\sum^{N}_{i,j=1}|i\rangle\langle j|\otimes |i\rangle\langle j|=|I\rangle\rangle\langle\langle I|$.
Furthermore, we have the swap operator $S=\sum_{i,j=1}^{N}|ij\rangle\langle ji|=\sum_{\alpha=1}^{N^{2}}E_{\alpha}\otimes E^{\dagger}_{\alpha}=\sum_{\beta=1}^{N^{2}}F_{\beta}\otimes F^{\dagger}_{\beta}$.\\

\item In fact, given two orthonormal bases
$\{E_{\alpha}\}_{\alpha=1}^{N^{2}}$ and
$\{F_{\alpha}\}_{\alpha=1}^{N^{2}}$ in $\mathcal{M}_{N}$, the
relation
\begin{eqnarray}
\mathscr{L}(\mathcal{M}_{N},\mathcal{M}_{N})\ni \Phi\longrightarrow
\Lambda_{\Phi} =\sum_{\alpha=1}^{N^{2}}\Phi(E_{\alpha})\otimes
F_{\alpha}\in \mathcal{M}_{N}\otimes\mathcal{M}_{N}
\end{eqnarray}
defines an isomorphism between
$\mathscr{L}(\mathcal{M}_{N},\mathcal{M}_{N})$ and
$\mathcal{M}_{N}\otimes\mathcal{M}_{N}$. The isomorphism is an
isometry:
\begin{eqnarray*}
\langle
\Lambda_{\Phi},\Lambda_{\Psi}\rangle&=&\langle\sum_{\alpha=1}^{N^{2}}\Phi(E_{\alpha})\otimes
F_{\alpha},\sum_{\beta=1}^{N^{2}}\Psi(E_{\beta})\otimes
F_{\beta}\rangle=\sum_{\alpha,\beta=1}^{N^{2}}\langle\Phi(E_{\alpha})\otimes
F_{\alpha},\Psi(E_{\beta})\otimes
F_{\beta}\rangle\\&=&\sum_{\alpha,\beta=1}^{N^{2}}\langle\Phi(E_{\alpha}),\Psi(E_{\beta})\rangle
\langle
F_{\alpha},F_{\beta}\rangle=\sum_{\alpha,\beta=1}^{N^{2}}\langle\Phi(E_{\alpha}),\Psi(E_{\beta})\rangle
\delta_{\alpha\beta}\\&=&\sum_{\alpha=1}^{N^{2}}\langle\Phi(E_{\alpha}),\Psi(E_{\alpha})\rangle
=\langle\Phi,\Psi\rangle; \mbox{i.e.}, \langle
\Lambda_{\Phi},\Lambda_{\Psi}\rangle=\langle\Phi,\Psi\rangle.
\end{eqnarray*}
\end{enumerate}
\end{remark}
\section{Best separable approximation for states}
In this section we recall the so-called optimal and the best separability approximation(OSA and BSA respectively).
Although the results below have been proven in \cite{Lewen,Sinisa}, we give the framework for our convenience. Other results involved can be found in \cite{Wellens}. In the method of BSA, for any density matrix $\rho$ there exist a ``optimal" separable matrix $\rho^{\ast}_{s}$ and ``optimal" non-negative scalar $\Lambda$ such that $\rho-\Lambda\rho^{\ast}_{s}\geq 0$. We describe these results involved that follow:

\begin{definition} A non-negative parameter $\Lambda$ is called maximal with respect to a (not necessarily
normalized) density matrix $\rho$, and the projection operator $P=|\psi\rangle\langle\psi|$ if $\rho-\Lambda P\geq0$, and
for every $\epsilon\geq0$, the matrix $\rho-(\Lambda+\epsilon)P$ is not positive definite.
\end{definition}
\begin{definition} A pair of non-negative $(\Lambda_{1},\Lambda_{2})$ is called maximal with respect to $\rho$ and a pair of projection operators $P_{1}=|\psi_{1}\rangle\langle\psi_{1}|, P_{2}=|\psi_{2}\rangle\langle\psi_{2}|$, if $\rho-\Lambda_{1}P_{1}-\Lambda_{2}P_{2}\geq0$, $\Lambda_{1}$ is
maximal with respect to $\rho-\Lambda_{2}P_{2}$ and to the projector $P_{1}$,$\Lambda_{2}$ is
maximal with respect to $\rho-\Lambda_{1}P_{1}$ and to the projector $P_{2}$ and and the sum $\Lambda_{1}+\Lambda_{2}$ is maximal.
\end{definition}
\begin{thrm} For any density matrix $\rho$ (separable, or not) and for any (fixed) countable set $V$ of product vectors belonging to the range of $\rho$, there exist $\Lambda(V)\geq0$ and a separable matrix
$$\rho^{\ast}_{s}(V)=\sum_{\alpha}\Lambda_{\alpha}P_{\alpha}$$
where each projector $P_{\alpha}$ is generated by some product vector in $R(\rho)$, and all $\Lambda_{\alpha}\geq0$, such that $\delta\rho=\rho-\Lambda\rho^{\ast}_{s}\geq0$, and that $\rho^{\ast}_{s}(V)$ provides the optimal separable approximation (OSA) to $\rho$ since $\mbox{Tr}(\delta\rho)$ is minimal or, equivalently, $\Lambda$ is maximal. There exists also the best separable approximation $\rho^{\ast}_{s}$ for which $\Lambda=\max_{V}\Lambda(V)$. Obviously, $\Lambda(V)\leq\Lambda(V')$ when $V'\subset V$.
\end{thrm}
\begin{thrm} Given the set $V$ of product vectors in the range $R(\rho)$ of $\rho$, the matrix $\rho^{\ast}_{s}=\sum_{\alpha}\Lambda_{\alpha}P_{\alpha}$ is the optimal separable approximation(OSA) of $\rho$ if:\\
1) all $\Lambda_{\alpha}$ are maximal with respect to $\rho_{\alpha}=\rho-\sum_{\alpha\neq\alpha}\Lambda_{\alpha'}P_{\alpha'}$, and to the projector $P_{\alpha'}$;\\
2) all pairs $(\Lambda_{\alpha},\Lambda_{\beta})$ are maximal with respect to $\rho_{\alpha\beta}=\rho-\sum_{\alpha\neq\alpha,\beta}\Lambda_{\alpha'}P_{\alpha'}$, and to the projection operators $(P_{\alpha},P_{\beta})$.
\end{thrm}
\begin{thrm}\textbf{(The uniqueness of the BSA)} Any density matrix $\rho$ has a unique decomposition
$\rho=\Lambda\rho_{s} + (1-\Lambda)\delta\rho$, where $\rho_s$ is a separable density matrix, $\delta\rho$ is a inseparable matrix with
no product vectors in its range, and $\Lambda$ is maximal.
\end{thrm}
\section{Best separable approximation for operations}
\noindent We cab define separable CPM; that is, $\Phi$ is separable if its action can be expressed in the form\\
\indent $\Phi(\rho)=\sum_{i=1}^{n}(S_{k}\otimes T_{k})\rho(S_{i}\otimes T_{k})^\dagger$,\\
for some integer $n$ and where $S_{k}$ and $T_{k}$ are operators acting on $\mathcal{H}_{A/B}$, respectively. Otherwise, we say that it is nonseparable.
Up to proportionality constant, separable maps are those that can be implemented using local operations and classical communication only.\\

Let us consider two systems, $A$ and $B$, spatially separated, each of them composed of two particles $(A_{1,2}$, and $B_{1,2})$. Let us consider a CPM $\Phi$ acting on systems $A_1$ and $B_1$, where $\Phi(\rho)=\sum_{k}M^{(k)}\rho(M^{(k)})^\dagger$ and $M^{(k)}$ acting on $\mathcal{H}_{A_{1}}\otimes\mathcal{H}_{B_{1}}$, where $M^{(k)}=S_{k}\otimes T_{k}$ for each index $k$. Now $\Phi$ induced another super-operator acting on $\mathcal{H}_{A}\otimes\mathcal{H}_{B}$ in the following sense:\\
\indent $\widetilde{\Phi}(X)=\sum_{k}(\widetilde{S_{k}}\otimes\widetilde{T_{k}})X(\widetilde{S_{k}}\otimes\widetilde{T_{k}})^{\dagger}$,\\
where $\widetilde{S_{k}}=S_{k}\otimes \mbox{id}^{(A_{2})}$ and $\widetilde{T_{k}}=T_{k}\otimes \mbox{id}^{(B_{2})}$, and $X$  acting on $\mathcal{H}_{A}\otimes\mathcal{H}_{B}$.\\
We are interested in whether this CPM can create ``nonlocal" entanglement between the systems $A$ and $B$. We define the operator $E_{A_{1}A_{2},B_{1}B_{2}}$ acting on $\mathcal{H}_{A}\otimes\mathcal{H}_{B}$, where $\mathcal{H}_{A}=\mathcal{H}_{A_{1}}\otimes\mathcal{H}_{A_{2}}$ and $\mathcal{H}_{B}=\mathcal{H}_{B_{1}}\otimes\mathcal{H}_{B_{2}}$, and $\dim(\mathcal{H}_{A_{i}})=\dim(\mathcal{H}_{B_{i}})=d$, as follows:\\
\indent $E_{A_{1}A_{2},B_{1}B_{2}}=(\Phi^{(A_{1}B_{1})}\otimes \mbox{id}^{(A_{2}B_{2})})(P_{A_{1}A_{2}}\otimes P_{B_{1}B_{2}})\equiv\widetilde{\Phi}(P_{A_{1}A_{2}}\otimes P_{B_{1}B_{2}})=\sum_{k}(\widetilde{S_{k}}\otimes\widetilde{T_{k}})(P_{A_{1}A_{2}}\otimes P_{B_{1}B_{2}})(\widetilde{S_{k}}\otimes\widetilde{T_{k}})^{\dagger}=\sum_{k}(\widetilde{S_{k}}P_{A_{1}A_{2}}\widetilde{S_{k}}^{\dagger}) \otimes(\widetilde{T_{k}}P_{B_{1}B_{2}}\widetilde{T_{k}}^{\dagger})$.\\
Here, $P_{A_{1}A_{2}}=|\Psi\rangle_{A_{1}A_{2}}\langle\Psi|$ with $|\Psi\rangle_{A_{1}A_{2}}=\frac{1}{\sqrt{d}}\sum_{m=1}^{d}|m\rangle_{A_{1}}\otimes|m\rangle_{A_{2}}$, and $P_{B_{1}B_{2}}=|\Psi\rangle_{B_{1}B_{2}}\langle\Psi|$ with  $|\Psi\rangle_{B_{1}B_{2}}=\frac{1}{\sqrt{d}}\sum_{\mu=1}^{d}|\mu\rangle_{B_{1}}\otimes|\mu\rangle_{B_{2}}$, where $\{|m\rangle:m=1,\ldots,d\}$ and $\{|\mu\rangle:\mu=1,\ldots,d\}$ are an orthonormal basis for $\mathcal{H}_{A_{1}/A_{2}}$ and $\mathcal{H}_{B_{1}/B_{2}}$ respectively. The map $\Phi$ is understood to act as the identity on the operators acting on $\mathcal{H}_{A_{2}}$ and $\mathcal{H}_{B_{2}}$. The operator $E$ has a clear interpretation since it is proportional to the density operator resulting from the operation $\Phi$ on systems $A_{1}$ and $B_{1}$ when both of them are prepared in a maximally entangled state with two ancillary systems, respectively. $E$ is called \emph{Choi matrix} for the bipartite super-operator $\Phi$, or the mapping $\Phi\rightarrow E(\Phi)$ is called \emph{Jamio{\l}kowski isomorphism} for the bipartite super-operator $\Phi$.\\

Now in general for $\Phi(\rho)=\sum_{k}M^{(k)}\rho(M^{(k)})^\dagger$, where $M^{(k)}=\sum_{mn,\mu\nu}M^{(k)}_{mn,\mu\nu}|m\rangle\langle n|\otimes|\mu\rangle\langle\nu|$, then $\Phi$ induced another super-operator like above as follows:\\
\indent $\widetilde{\Phi}(X)=\sum_{k}\widetilde{M^{(k)}}X\left(\widetilde{M^{(k)}}\right)^{\dagger}$,\\
where $\widetilde{M^{(k)}}=\sum_{mn,\mu\nu}M^{(k)}_{mn,\mu\nu}\widetilde{|m\rangle\langle n|}\otimes\widetilde{|\mu\rangle\langle\nu|}$; and $\widetilde{|m\rangle\langle n|}=|m\rangle\langle n|\otimes\mbox{id}^{(A_{2})}$ and $\widetilde{|m\rangle\langle n|}=|\mu\rangle\langle \nu|\otimes\mbox{id}^{(B_{2})}$.
\begin{eqnarray*}
E_{A_{1}A_{2},B_{1}B_{2}}&=&(\Phi^{(A_{1}B_{1})}\otimes \mbox{id}^{(A_{2}B_{2})})(P_{A_{1}A_{2}}\otimes P_{B_{1}B_{2}})\\
&\equiv&\widetilde{\Phi}(P_{A_{1}A_{2}}\otimes P_{B_{1}B_{2}})=\sum_{k}\widetilde{M^{(k)}}(P_{A_{1}A_{2}}\otimes P_{B_{1}B_{2}})\left(\widetilde{M^{(k)}}\right)^{\dagger}\\
&=&\sum_{k}\sum_{mm'nn',\mu\mu'\nu\nu'}M^{(k)}_{mn,\mu\nu}[M^{(k)}_{m'n',\mu'\nu'}]^{\ast}\\
&&(|m\rangle\langle n|\otimes\mbox{id}^{(A_{2})})P_{A_{1}A_{2}}(|m'\rangle\langle n'|\otimes\mbox{id}^{(A_{2})})^{\dagger}\otimes (|\mu\rangle\langle\nu|\otimes\mbox{id}^{(B_{2})})P_{B_{1}B_{2}}(|\mu'\rangle\langle \nu'|\otimes\mbox{id}^{(B_{2})})^{\dagger}\\
&=&\sum_{k}\sum_{mm'nn',\mu\mu'\nu\nu'}M^{(k)}_{mn,\mu\nu}[M^{(k)}_{m'n',\mu'\nu'}]^{\ast}\\
&&(|m\rangle\langle n|\otimes\mbox{id}^{(A_{2})}\otimes|\mu\rangle\langle\nu|\otimes\mbox{id}^{(B_{2})})(P_{A_{1}A_{2}}\otimes P_{B_{1}B_{2}})(|m'\rangle\langle n'|\otimes\mbox{id}^{(A_{2})}\otimes|\mu'\rangle\langle \nu'|\otimes\mbox{id}^{(B_{2})})^{\dagger}.
\end{eqnarray*}
If we define $\mbox{vec}(M^{(k)})=\sum_{mn,\mu\nu}M^{(k)}_{mn,\mu\nu}|m\rangle|n\rangle\otimes|\mu\rangle|\nu\rangle=\sum_{mn,\mu\nu}M^{(k)}_{mn,\mu\nu}|mn\rangle\otimes|\mu\nu\rangle=\sum_{mn,\mu\nu}M^{(k)}_{mn,\mu\nu}|mn\mu\nu\rangle$, then
$\sum_{k}\mbox{vec}(M^{(k)})\mbox{vec}(M^{(k)})^{\dagger}=E_{A_{1}A_{2},B_{1}B_{2}}$. If $M^{(k)}=A_{k}\otimes B_{k}$, then $\mbox{vec}(M^{(k)})=|A_{k}\rangle\rangle|B_{k}\rangle\rangle$.
\begin{prop} If $\Phi$ is a quantum operation on a bipartite quantum system, then $\Phi$ is separable if and only if its dynamical matrix $D_{\Phi}$ is separable.
\end{prop}
\begin{proof} By the definition of separable quantum operation,
$\Phi(\rho)=\sum_{i}(A_{i}\otimes B_{i})\rho(A_{i}\otimes
B_{i})^{\dagger}$ when $\Phi$ is separable. Now the dynamical matrix
for the separable operation $\Phi$ is
$D_{\Phi}=\sum_{i}\mbox{vec}(A_{i}\otimes B_{i})
\mbox{vec}(A_{i}\otimes B_{i})^{\dagger}=\sum_{i}|A_{i}\rangle\rangle\langle\langle
A_{i}|\otimes |B_{i}\rangle\rangle\langle\langle B_{i}|$.
\end{proof}
\begin{definition} Given quantum operation $\Phi$ on bipartite quantum system $\mathcal{H}_{1}\otimes\mathcal{H}_{2}$ with $\mbox{dim}\mathcal{H}_{1}=\mbox{dim}\mathcal{H}_{2}=N$,
the dynamical matrix $D_{\Phi}$ for $\Phi$ can decomposed as $D_{\Phi}=\lambda D_{s}+(1-\lambda)D_{e}$ in terms of the BSA decomposition for state.
Then the separable operation $\Phi_{BSA}$ determined by $\lambda D_{s}$ is called \emph{best separable operation approximation} for $\Phi$.
$\Phi_{ENT}\equiv\Phi-\Phi_{BSA}$ is called \emph{pure entanglement-produced operation part} for $\Phi$.
\end{definition}
\begin{remark} If there is another decomposition $D_{\Phi}=D'_{s}+D'_{e}$ for which $D'_{s}$ is just separable, then: $\lambda D_{s}-D'_{s}\geq0$ by the uniqueness of the BSA. Thus the decomposition $\Phi=\Phi_{BSA}+\Phi_{ENT}$ is unique. \\~\\
By operator-sum representation theorem,
$\Phi(\rho)=\sum_{i\in\mathbb{F}}F_{i}\rho F^{\dagger}_{i}=\sum_{j\in\mathbb{G}}G_{j}\rho G^{\dagger}_{j},$
where $\max(|\mathbb{F}|,|\mathbb{G}|)\leq N^4$. Let $$\mathbb{I}=\{i\in\mathbb{F}:F_{i}=A_{i}\otimes B_{i}\},\mathbb{J}=\{j\in\mathbb{G}:G_{j}=C_{j}\otimes D_{j}\}.$$
Write
$\Upsilon(\rho)=\sum_{i\in\mathbb{I}}F_{i}\rho F^{\dagger}_{i}$ and $\Psi(\rho)=\sum_{j\in\mathbb{J}}G_{j}\rho G^{\dagger}_{j}$;$\Upsilon'(\rho)=\sum_{i\in\mathbb{F}\setminus\mathbb{I}}F_{i}\rho F^{\dagger}_{i}$ and $\Psi'(\rho)=\sum_{j\in\mathbb{G}\setminus\mathbb{J}}G_{j}\rho G^{\dagger}_{j}$.
\end{remark}
\begin{thrm} $\Upsilon=\Psi=\Phi_{BSA}.$
\end{thrm}
\begin{proof} Apparently, $D_{\Phi}=D_{\Upsilon}+D_{\Upsilon'}$,
where $D_{\Upsilon}$ is separable since $\Upsilon$ is separable
operation. Hence it follows from the uniqueness of the BSA that
$\lambda D_{s}-D_{\Upsilon}\geq0$. If, otherwise, $\lambda
D_{s}-D_{\Upsilon}>0$, then
$D_{\Phi_{BSA}-\Upsilon}=D_{\Phi_{BSA}}-D_{\Upsilon}>0$, that is,
$\Phi_{BSA}-\Upsilon$ is CP map and separable, so
$D_{\Phi}=[D_{\Upsilon}+D_{\Phi_{BSA}-\Upsilon}]+[D_{\Upsilon'}-D_{\Phi_{BSA}-\Upsilon}]$,
where $D_{\Upsilon'}-D_{\Phi_{BSA}-\Upsilon}>0$, contradict with the
fact that there is no factorizing operational element for
$\Upsilon'$. Therefore $\lambda D_{s}-D_{\Upsilon}=0$, equivalently,
$\Upsilon=\Phi_{BSA}$. The theorem is proved.
\end{proof}

\end{document}